\documentclass[11pt]{article}

\usepackage[margin=1in, letterpaper]{geometry}
\usepackage{setspace}
\usepackage[utf8]{inputenc}
\usepackage[numbers]{natbib}
\usepackage[linesnumbered,ruled,vlined]{algorithm2e}
\SetKwProg{upon}{upon}{}{}
\newcommand{\codefont}[1]{\textup{\texttt{#1}}}

\usepackage{xcolor} \SetCommentSty{(\color{blue}\itshape)} 

\SetCommentSty{mycommfont}
\SetKwProg{Proc}{Procedure}{}{}

\usepackage{array}
\newcolumntype{P}[1]{>{\centering\arraybackslash}p{#1}}

\usepackage{amsmath,amsthm,amssymb}

\usepackage{hyperref}
\usepackage{cleveref}

\newtheorem{theorem}{Theorem}
\newtheorem{lemma}[theorem]{Lemma}
\newtheorem{corollary}[theorem]{Corollary}
\newtheorem{definition}[theorem]{Definition}

\newtheorem{observation}[theorem]{Observation}

\DeclareMathOperator{\suc}{succ}
\DeclareMathOperator{\pre}{pred}

\DeclareMathOperator{\indeg}{indeg}
\DeclareMathOperator{\outdeg}{outdeg}

\usepackage{nicefrac}
\usepackage{soul}
\newcommand{\E}{\mathbb{E}}

\usepackage{mathtools}

\usepackage{tikz}
\usetikzlibrary{arrows.meta}
\usetikzlibrary{decorations.pathreplacing, positioning, calc}

\title{Supervised Distributed Computing:  Efficiency and Robustness under a Majority of Adversarial Workers
}
\author{John Augustine{\small ,} Henning Hillebrandt{\small ,} Manish Kumar{\small ,}\\ Christian Scheideler, {\small and} Julian Werthmann}
\date{}

\begin{document}

\maketitle

\begin{abstract}
    We consider a recently proposed \emph{supervised distributed computing} paradigm \cite{augustine2025supervised} that extends and refines the standard master-worker paradigm for parallel computations.
    In this paradigm, there is a supervisor, a source, a target, and a collection of workers. The distributed computation is given as an acyclic task graph that is known to the supervisor. The source initially stores the input and the target is supposed to store the output of the computation. 
    The individual tasks of the computation are supposed to be executed by the workers under the guidance of the supervisor.
    The source, target and supervisor are assumed to be reliable, while a $\beta$-fraction of the workers might be adversarial, for some $\beta \in [0,1)$. This covers, for example, the case where a supervisor has to work with untrusted volunteers.
    In the standard master-worker approach, the master checks whether the workers correctly execute the assigned tasks, creating a severe bottleneck, whereas in the supervised approach, the supervisor outsources this checking to the workers.
    Prior to this work, only supervised solutions were known for the case that $\beta$ is a sufficiently small constant. 
    We show that robust and efficient supervised solutions are possible for \emph{any} constant $\beta<1$ while the expected work for the honest workers is close to a \emph{single} execution per task, given that there is a lightweight verification mechanism that allows honest workers to check the correctness of task outputs, which is significantly better than all robust master-worker as well as peer-to-peer approaches known so far.
\end{abstract}


\section{Introduction}\label{sec: introduction}

In the standard master-worker paradigm, there is a master and a collection of workers. The master is responsible for assigning tasks to the workers and collecting their outputs while the workers are supposed to execute the tasks in a correct and timely manner. Countermeasures against workers who do not deliver their outputs in a timely manner are relatively simple: Wait until a timeout has been reached, and then assign the task to another worker. However, checking the correctness of the output is more challenging. A standard approach has been to assign the same task to two workers, and if their outputs disagree, to assign it to a third worker or more until a majority of workers agree on the output. This has turned out to be effective in practice against isolated adversarial workers (e.g., SETI@home \cite{seti2001}), but is not effective if there is a constant fraction of cooperating adversarial workers. Various alternative approaches may be used such as challenging workers with tasks whose correct outputs are already known to identify and blacklist adversarial workers, but in this case, computations performed by honest workers are wasted. Thus, ideally, appropriate verification mechanisms are available that allow the master to check the correctness of an output \emph{without} consulting other workers. Such verification mechanisms are known for various problems and have been studied in the context of \emph{certifying algorithms}~\cite{McConnellMNS11} and interactive proofs~\cite{feigenbaum1992overview}. In the maxflow problem, for example, the master may not just expect a maxflow solution from the worker but also a mincut, which must have the same values to be correct due to the famous maxflow-mincut theorem. For distributed computations, it can be more challenging to identify lightweight verification mechanisms at the task level, but as was demonstrated in \cite{augustine2025supervised}, this is possible, for example, for distributed matrix multiplication and sorting.

Certainly, requiring the master to perform these verifications puts a significant burden on the master and therefore severely limits the scalability of the master-worker approach. To avoid this problem, the supervised distributed computing paradigm was recently proposed, which shifts the burden of verifying the outputs of task computations to the workers so that the supervisor is not involved in any I/O activities anymore. Instead, the supervisor just focuses on assigning tasks to workers and introducing them to each other to make sure that the computation proceeds in a controlled manner. To ensure that there is a chance for the computation to be performed correctly, a critical assumption of that approach is to assume that the source, which stores the input of the distributed computation, and the target, which is supposed to collect the final outputs of the computation, are reliable. In that case, the input will be delivered correctly to the workers responsible for the initial tasks, and when the target tells the supervisor that all final outputs have been correctly delivered to it, the supervisor knows for sure that the computation was successful.

Prior to our work, only efficient supervised distributed computing solutions were known for the case where the fraction of adversarial workers is a sufficiently small constant \cite{augustine2025supervised}. In this work, we show that there are supervised approaches that are efficient and robust against \emph{any} constant fraction of adversarial workers less than 1. Before we state our concrete contributions, we formally introduce the supervised distributed computing paradigm and give an overview of related work.

\subsection{Supervised Distributed Computing Paradigm}\label{sec: model}

In the supervised distributed computing paradigm, there are four roles: a \emph{supervisor}, a \emph{source}, a \emph{target}, and a collection of \emph{workers}. The term \emph{supervisor} was suggested instead of \emph{master} to highlight its lightweight scheduling role. The source, target, and supervisor are assumed to be \emph{reliable} in the sense that they operate on time and are not subject to crashes or adversarial behavior.

Given an instance $I$ of a computational problem, $I$ is initially stored in the source, and the supervisor's goal is to orchestrate the workers so that, in the end, a correct solution $S(I)$ to $I$ is stored in the target. Note that the supervisor, source, and target are abstractions. Depending on the application, these roles may overlap or may be realized by multiple entities. For example, if $S(I)$ is small (e.g., a simple ``yes/no'' answer), it might be convenient for the supervisor and target to be the same entity. In decentralized environments with a sufficiently small fraction of adversarial peers, a reliable supervisor might be emulated by a quorum of peers, and a source may consist of a collection of repositories with read-only access.

We assume that the distributed computation for some instance $I$ can be represented as a \emph{directed acyclic task graph} $G = (V, E)$ that is known to the supervisor.\footnote{Example task graphs for matrix multiplication and sorting have been presented in \cite{augustine2025supervised}.} Each node $v \in V$ represents a \emph{task}, and for every edge $(u,v) \in E$, $u$ is called a \emph{predecessor} of $v$ and $v$ is called a \emph{successor} of $u$. 
We denote the set of successors of a task $v$ in $G$ as $\suc_G(v)$ and the set of predecessors in $G$ as $\pre_G(v)$.
For every task $v \in V$, $v$ needs the outputs of all predecessors to be executable. Nodes without incoming edges represent \emph{initial tasks}; they receive their inputs directly from the source. Nodes without outgoing edges represent \emph{final tasks}; they are supposed to send their outputs to the target, where the final solution $S(I)$ is assembled. For our analysis, we use $n = |V|$ to denote the number of tasks and $D$ to denote the \emph{depth} of $G$, i.e., the length of the longest directed path in the graph.
Furthermore, for each task $v \in V$, we use $D(v) \in \{1, \dots, D\}$ to denote the depth of $v$ in $G$, i.e., the length of the longest directed path from an initial task to $v$.

\subsubsection{Worker Sampling and Adversarial Model}

We assume that the supervisor has access to a black-box sampling mechanism that, when asked for a worker, returns an \emph{adversarial worker} with probability at most $\beta \in [0,1)$, and otherwise an \emph{honest worker}.
The sampling mechanism ensures that only \emph{available} workers are picked, and an honest worker only declares itself available if it is not currently busy with a task. The set of workers may change arbitrarily over time, but we assume that there are always sufficiently many available honest workers so that the sampling mechanism can work as desired.
$\beta$ is known to the supervisor, but this is not a severe constraint since the supervisor may simply choose a higher $\beta<1$ if a correct solution $S(I)$ cannot be computed.
The adversary is omniscient, i.e., it has complete knowledge of the system's past and present state and can coordinate all adversarial workers. However, it cannot predict the supervisor's future random choices or the future choices of the sampling mechanism. Honest workers, once selected, are assumed to remain accessible and function correctly for as long as the output of the task assigned to them is needed for the execution of $G$ (which depends on the scheduling strategy of the supervisor).

\subsubsection{Task Executions}

For simplicity, we assume that the time progresses in \emph{synchronous rounds}. Each round is sufficiently long to ensure that the following sequence of events can take place: (1) The supervisor samples the desired number of workers for tasks that it wants to get executed in that round, (2) the supervisor introduces workers to each other so that formerly sampled workers can forward their outputs to the corresponding newly sampled workers, (3) the newly sampled workers check whether the received outputs are correct and complete,\footnote{Workers can discard any received outputs that do not adhere to the protocol, including messages that are too large.} and (4) if so, they may execute the task assigned to them.

The assumption of synchronous rounds is reasonable if all tasks require roughly the same computational effort and the time required for receiving outputs and verifying their correctness is negligible compared to executing the task. The most critical issue is coming up with a lightweight verification mechanism for the received outputs, which has been shown to be possible for certain problems and is an active area of research, but here we just assume that such verification mechanisms are available and focus instead on the scheduling part.\footnote{The assumption is reasonable, in particular due to the existence of class NP of problems that are difficult to compute but easy to verify.}

\subsubsection{Efficiency and Robustness Goals}

We consider a supervised solution to be \emph{efficient} if the following goals can be met:
\begin{itemize}
    \item The expected communication work of the source is linear in the size of the input $I$.
    \item The expected communication work of the supervisor is linear in the size of $G$ and independent of the size of the input $I$ and its solution $S(I)$.
    \item The expected communication and verification work of the honest workers is in $O(d)$, where $d$ is a bound on the in- and outdegree of $G$.
    \item The expected total computation work of the honest workers is asymptotically equal to the work required for executing all tasks in $G$.
    \item The expected runtime of the computation, i.e., the number of rounds it takes to terminate, is asymptotically equal to the depth $D$ of the task graph.
    \item The expected communication work of the target is linear in the size of the solution $S(I)$.
\end{itemize}
These requirements explicitly rule out trivial scheduling strategies. For example, assigning the entire computation to a single worker would result in a runtime of $O(|V|)$, which can be significantly larger than $D$. Similarly, involving the supervisor in handling the input and output data would make its work dependent on $|I|$ or $|S(I)|$, which can be much larger than the size of $G$, thereby creating a bottleneck.

We call a supervised solution \emph{robust} if, despite the presence of adversarial workers, eventually a correct solution can be assembled at the target.

\section{Related Work}\label{sec: related_work}

The master-worker paradigm has been around for decades and is still widely used in practice. Approaches following that paradigm can be found in volunteer-based computing (e.g., PrimeNet and SETI@home) as well as data processing in large computational clusters (e.g., Google's MapReduce paradigm \cite{DeanG04}). While in the basic master-worker paradigm, tasks are typically independent, the MapReduce paradigm can cover situations where a task may spawn several subtasks (the Map operation) whose outputs are collected (via the Reduce operation) to assemble the final output of the task.

Starting with the work of Fern\'andez, Georgiou, L\'opez and Santos~\cite{FernandezGLS05}, and followed by Fer\-n\'an\-dez, L\'opez, Santos and Georgiou \cite{FernandezLSG06}, various algorithmic techniques have been explored for the case that the tasks are independent of each other and the tasks need to be computed by workers that are potentially malicious (e.g., \cite{ChristoforouAGM14,KonwarRS15}). A standard assumption in these papers is that there is a fixed set of $n$ workers and a set of $t$ tasks, where $n \le t$. The computation proceeds in synchronous rounds. In each round, the master can assign any collection of tasks to the workers, and during a round, each worker can compute only one task from the master and report the result back to it. The goal of the master is to accept only correct answers of every task, w.h.p.\footnote{We say that an event occurs with high probability (w.h.p.), if it occurs with probability at least $1-n^{-c}$ for a constant $c$.} Techniques that have been explored to achieve that are voting (a majority of workers have to agree on the output), challenges (where tasks with known outputs are given to workers to identify malicious ones), or combinations of these. Appropriate strategies for the master based on these techniques were presented in \cite{ChristoforouAGM14,KonwarRS15} that, for the case of binary outputs, achieve the stated goal if the probability of malicious or faulty behavior is less than $1/2$, and for scenarios where more outputs are allowed, this bound may even be lower. Furthermore, a lower bound of $\Omega(t \log n)$ for the total work was shown for the case that an $f$-fraction, $0<f<1/2$, of the workers may return incorrect results with probability $0<p<1/2$, while the master does not have a priori knowledge of $f$ and $p$ \cite{KonwarRS15}. Reputation-based (e.g., \cite{ChristoforouAGM16,SonnekNCW06}) and game-theoretical approaches (e.g., \cite{AntaGMP16,ChristoforouAKN16}) have also been considered in this context, but these are beyond the scope of this paper.

In the works above, the master does all the checking. Goodrich~\cite{Goodrich08} investigated a method to prevent cheating that uses the workers to check each others' outputs. As the main result, he shows that three rounds and a replication factor of 6 suffice to identify all false outputs among $t=n$ tasks, so long as there are at most 5\,\% of malicious participants. He actually calls the master a supervisor, and his approach appears to be applicable to the supervised distributed computing paradigm, but it requires the number of malicious participants to be sufficiently small to work.

The supervised distributed computing model considered here was introduced by Augustine, Scheideler and Werthmann~\cite{augustine2025supervised}. They studied two classes of task graphs: paths (as a warm-up) and arbitrary directed acyclic graphs (DAGs). For the path case, they showed that when $\beta \leq 1/12$, their scheduling approach terminates correctly in a linear number of rounds, w.h.p.\ (here, $n$ denotes the number of tasks of the task graph). Moreover, for any adversarial strategy, the source and target send and receive the solution only a constant number of times in expectation, and the total computational and verification work performed by the workers is within a constant factor of optimal. 
They then generalized their path strategy to DAGs and showed that as long as $\beta \leq (1/(2(2d+1)))^{2+\epsilon}$, for an arbitrary constant $\epsilon > 0$, where $d$ denotes the maximum indegree of the task graph, the computation finishes within $O(D+ \log n)$ rounds, w.h.p.

Research on verification in adversarial environments can be traced back to the pioneering work of Blum, Evans, Gemmell, Kannan, and Naor~\cite{BEGKN91}, who studied how to check the correctness of operations performed on a data structure controlled by an adversary. 
Their work initiated a significant amount of research on authenticated data structures (see \cite{Tamassia03} for a survey), which typically involves a trusted source, an untrusted server, and a client. 
The paradigm of certifying algorithms~\cite{AlkassarBMR14} further explores the separation between computation and verification. In this setting, an algorithm outputs not only a solution but also a witness whose correctness can be efficiently verified by a checker. Certifying algorithms have been developed for a variety of sequential problems, including graph problems and geometric settings (e.g., \cite{McConnellMNS11}). 

There is a large body of work on cryptographic approaches to check the correctness of outputs (see \cite{DBLP:conf/dlt/CrescenzoKKS22} for a survey). This line of work is commonly formalized using interactive proof systems~\cite{feigenbaum1992overview}, where a powerful prover convinces a weaker verifier of the correctness of a computation. More recent developments include non-interactive arguments such as SNARGs and SNARKs~\cite{chiesa2014succinct,nitulescu2020zk}, which aim at succinct verification and have found applications in blockchain systems. Supervised distributed computing solutions (ours included) are designed to take advantage of such efficient verification tools to solve large-scale parallel and distributed computing problems in the presence of Byzantine participants. 

\subsection{Our Contributions}

We present a supervised distributed computing solution that can tolerate a $\beta$-fraction of adversarial workers for any constant $\beta < 1$ and any directed acyclic task graph. This is a significant improvement over the DAG scheduling approach of \cite{augustine2025supervised}, which only preserves efficiency for $\beta \in O((\log d)/d^2)$, where $d$ is a bound on the maximum in- and outdegree of the task graph $G$ (cf.~\Cref{sec:bound_and_infeasibility} for the proof). Since honest workers can be a minority, all majority-based solutions will be unhelpful, even if we assign a large number of workers per task.  

Our supervised strategy assigns $\Theta(\log n)$ many workers to each task $v$.
Rather than letting all workers of a task run in parallel, the supervisor schedules them sequentially, assigning one worker per round across $\Theta(\log n)$ consecutive rounds.
The scheduling is pipelined so that the first worker is assigned to task $v$ in round $\Theta(D(v) \cdot \log(d) \cdot \log\log n)$.
When a worker is assigned to a task, it first checks the outputs of $O(\log d \cdot \log \log n)$ previously assigned workers of that task (if they exist, respectively).
If any of them provides a correct output, the worker is done, having learned the result of its task without performing the computation itself.
Otherwise, it contacts $O(\log d \cdot \log \log n)$ previously assigned workers of each preceding task. 
If it gets a correct output for every preceding task, it executes the task assigned to it and is done. 
Otherwise, it fails.

As it turns out, our solution is highly efficient in the sense that each task is executed by only $1+o(1)$ honest workers on expectation, and therefore, the total expected number of computations performed by honest workers is $n(1+o(1))$, where $n$ is the number of tasks. Moreover, owing to the way the workers are assigned to the tasks, the distributed computation finishes in $O(D \log d \cdot \log \log n + \log n)$ rounds, w.h.p.\ (see Theorem~\ref{th:main_theorem_dags}).

To show that the computation finishes successy, w.h.p., we use a complex witness argument that is able to identify a significant amount of places where adversarial workers were chosen in the task graph whenever the computation fails. This then allows us to show that the probability for such a witness to be found when using the random sampling mechanism is polynomially small, thereby implying that our solution is correct w.h.p. 


\section{Handling a Majority of Adversarial Workers}
\label{sec:many_workers}

    Before we begin describing our algorithm, we remark that \cite{augustine2025supervised} shows how to transform any task graph $G = (V,E)$ into a leveled network of the same depth, i.e., a graph where the tasks are divided into levels and edges only connect tasks in adjacent levels.
    The key idea is to assign each task $v \in V$ to level $D(v)$.
    All edges that connect tasks of non-adjacent levels are then replaced by paths consisting of one additional task for each level in between.
    Note that this transformation is performed by the supervisor before any distributed computation starts. Therefore, we will w.l.o.g. only consider the case where the task graph is a leveled network. 

    To motivate our new scheduling approach, we start with a simple straw man algorithm that can handle an arbitrary $\beta$-fraction of adversarial workers for any constant $\beta < 1$, 
    but comes with significant downsides.
    Suppose the task graph $G = (V,E)$ is a leveled network of depth $D$.
    In each round $t = 1, \dots, D$, the supervisor assigns $\gamma$ many workers to each task $v \in V$ with $D(v) = t$ and instructs all workers assigned to predecessors of $v$ (or the source if $t=1$) to send their output to all workers assigned to $v$.
    Every honest worker then verifies all the outputs it received and, if it received at least one correct output for each preceding task, computes the task it was assigned to.
    If we set $\gamma = \Theta(\log n)$, one can prove that every task will be assigned at least one honest worker, w.h.p., which implies that the target will learn the correct output and the computation succeeds.
    Although this algorithm is time-optimal, every honest worker needs to communicate with $\Theta(d\log n)$ other workers and needs to perform $\Theta(d\log n)$ many verifications, where $d\geq\max\{\indeg(G),\outdeg(G)\}$ is an upper bound on the degree of the task graph. 
    While $\Omega(d)$ is a natural lower bound on the number of verifications and communication required per worker, the logarithmic dependency on $n$ severely limits the algorithm's performance in practice.
    Furthermore, each task gets executed by $\Theta(\log n)$ honest workers in expectation, which wastes computational resources.

    \subsection{Algorithm Description}
    We reduce the communication and verification costs of the straw man algorithm by letting each worker send its output only to a small subset of the workers assigned to each succeeding task in the following way.
    Instead of immediately assigning $\gamma$ workers to a single task per round, the supervisor only assigns up to one worker per task per round, but may assign workers to multiple tasks at once.
        Specifically, the supervisor assigns workers to each task $v$ of depth $D(v)$ between rounds $t_{\mathrm{min}}(v) := (D(v)-1) \cdot \delta + 1$ and $t_{\mathrm{max}}(v) := (D(v)-1) \cdot \delta + \gamma$, where $\delta \geq 1$ is some fixed parameter that determines the time offset between tasks.
    Hence, in each round $t = 1, \dots, \gamma + (D - 1) \cdot \delta$, the supervisor assigns exactly one worker to each task $v$ with $t \in [t_{\mathrm{min}}(v), t_{\mathrm{max}}(v)]$.
    For every worker $w$ assigned to task $v$ in round $t$, the supervisor first instructs all workers assigned to task $v$ in rounds $t' \in [\max\{t-2\delta, t_{\mathrm{min}}(v)\}, t-1]$ to send their output for task $v$ to $w$.
    If it is honest, $w$ verifies these outputs one by one until it succeeds in verifying an output or fails to verify all outputs.
    If it succeeds in verifying an output, it does not need to do any further verifications and no computation, but is required to remain available for $2\delta$ additional rounds to provide $v$'s output to its successors.
        Otherwise, $w$ informs the supervisor about its failure to verify.
    The supervisor then instructs all workers assigned to any preceding task $v' \in \pre_G(v)$ in rounds $t' \in [\max\{t-2\delta, t_{\mathrm{min}}(v')\}, t-1]$ to send their output for task $v'$ to $w$.
    Again, if it is honest, $w$ verifies these outputs one by one until it succeeds for one of them or fails for all of them. 
    If $w$ receives a correct output for each of $v$'s predecessors, $w$ computes the output of $v$ and remains in the system for $2\delta$ further rounds to provide it to its successors.
    Otherwise, it failed to compute an output and remains in the system for $2\delta$ rounds to provide an error message to its successors.
    We present pseudocode for the supervisor in \Cref{alg:manyworkers_supervisor} and for the workers in \Cref{alg:manyworkers_worker}, both located in \Cref{sec:pseudocode}.

    With this approach, every honest worker needs to communicate with only $O(d\delta)$ other workers and needs to perform only $O(d\delta)$ verifications.
    Furthermore, since workers can forward their output to (some) other workers assigned to the same task, the expected number of task executions by honest workers is reduced.
    However, the probability that the computation succeeds now depends on $\delta$.
    We show in the remainder of the section that the algorithm still succeeds with high probability if we set $\delta = \Theta(\log_{\nicefrac{1}{\beta}}(d) \cdot \log_{\nicefrac{1}{\beta}}\log n)$ and $\gamma = \Theta(\log n)$.

    \subsection{Preliminaries for the Analysis}

    We model the behavior of the algorithm on a task graph $G = (V,E)$ as a \emph{worker graph} $G_W = (V_W, E_W)$, in which each node $(v,t) \in V_W$ represents a worker that was assigned to task $v \in V$ in round $t$ and each edge $(w,w') \in E_W$ represents that worker $w$ might send data to worker $w'$.
    An example of a worker graph is depicted in \Cref{fig:worker_graph}.
    
    \begin{definition}[Worker Graph]
        For a task graph $G=(V,E)$, we define the worker graph $G_W=(V_W,E_W)$.         For each task $v$ of depth $D(v)$, $V_W$ contains a node $(v,t)$ for each time $t$ at which the supervisor assigns a worker to $v$, i.e., $V_W=\{(v,t)\mid v\in V \land t\in\{(D(v)-1) \cdot \delta+1,\dots,(D(v)-1)\cdot \delta+\gamma\}\}$. For $w,w'\in V_W$, $E_W$ contains an edge from $w$ to $w'$ if $w$ might send data to $w'$, i.e., $E_W = E_W^{(1)} \cup E_W^{(2)}$ with 
        \[
            E_W^{(1)} = \{((v,t),(v',t')) \in V_W\times V_W \mid v'=v \land t'\in\{t+1,\dots,t+2\delta\}\}
        \]
        and 
        \[
            E_W^{(2)} = \{((v,t),(v',t')) \in V_W\times V_W \mid v'\in\suc_G(v) \land t'\in\{t+1,\dots,t+2\delta\}\}\text{.}
        \]
        We refer to the nodes of $G_W$ as workers and call a node $(v,t)\in V_W$ an initial worker, if $v$ is an initial task of $G$, and a final worker, if $v$ is a final task of $G$.
                            \end{definition}

    \begin{figure}[ht!]
        \begin{center}
            \begin{tikzpicture}
                \node at (2,-3){};
                \node[fill=black!60,circle] (0-0) at (0,0) {};
                \node[fill=black!60,circle] (1-1) at (1.5,-1) {};
                \node[fill=black!60,circle] (1-0) at (1.5,0) {};
                \node[fill=black!60,circle] (0-1) at (0,-1) {};
                \node[fill=black!60,circle] (0-2) at (0,-2) {};
                \node[fill=black!60,circle] (1-2) at (1.5,-2) {};
                \draw[thick,black!30!,->] (0-0) -- (0-1);
                \draw[thick,black!30!,->] (1-0) -- (0-1);
                \draw[thick,black!30!,->] (1-0) -- (1-1);
                \draw[thick,black!30!,->] (0-1) -- (0-2);
                \draw[thick,black!30!,->] (1-1) -- (0-2);
                \draw[thick,black!30!,->] (1-1) -- (1-2);
            \end{tikzpicture}
            \begin{tikzpicture}[scale=0.4]
                \def\layers
                {
                    {{0,0,0,0,0,0,0,0,1,0,0,0,0},{0,0,0,1,0,1,0,0,0,0,0,0,0}},
                    {{0,0,0,1,0,0,0,0,0,0,2,0,0},{0,0,0,0,1,0,0,1,0,0,0,0,0}},
                    {{0,0,0,0,0,0,0,2,0,0,0,0,0},{0,0,0,0,0,0,0,0,0,0,0,0,0}},
                }
                \def\numoffset{1}
                \def\numsuccs{2}
                \def\numassign{13}

                \def\xoffset{2}
                \def\yscale{3}

                \foreach \layer [count=\y from 0] in \layers {
                    \pgfmathparse{dim({\layer})}
                    \let\layersize\pgfmathresult
                    \pgfmathparse{-(\layersize-1)/2}
                    \let\xstart\pgfmathresult
                    \foreach \task [count=\xx from 0] in \layer{
                        \pgfmathparse{\y+(\numassign+\xoffset)*(\xstart+\xx)}
                        \let\x\pgfmathresult
                        \fill[black!5,rounded corners] (\x,-\y*\yscale-0.25) rectangle (\x+\numassign,-\y*\yscale+0.25);
                        \foreach \worker [count=\w] in \task {
                            \node[circle,fill=black,inner sep=1.4pt] (\y-\xx-\w) at (\x+\w-0.5,-\y*\yscale) {};
                        }
                    }
                }

                \foreach \y/\x/\xx in {0/0/0,0/1/0,0/1/1,1/0/0,1/1/0,1/1/1} {
                    \pgfmathparse{int(\y+1)}
                    \let\yy\pgfmathresult
                    \foreach \w in {1,...,\numassign}{
                        \pgfmathparse{int(\w+1-\numoffset)}
                        \let\wmin\pgfmathresult
                        \pgfmathparse{int(\w+\numsuccs-\numoffset)}
                        \let\wmax\pgfmathresult
                        \foreach \ww in {\wmin,...,\wmax}{
                            \ifnum\ww>\numassign\else
                                \draw[black!10] (\y-\x-\w) -- (\yy-\xx-\ww);
                            \fi
                        }
                    }
                }

                \foreach \layer [count=\y from 0] in \layers {
                    \foreach \task [count=\x from 0] in \layer{
                        \foreach \worker [count=\w] in \task {
                            \pgfmathparse{int(\w+1)}
                            \let\wmin\pgfmathresult
                            \pgfmathparse{int(\w+\numsuccs)}
                            \let\wmax\pgfmathresult
                            \foreach \ww in {\wmin,...,\wmax}{
                                \ifnum\ww>\numassign\else
                                    \draw[black!10] (\y-\x-\w) to[bend left] (\y-\x-\ww);
                                \fi
                            }
                        }
                    }
                }

                \draw[->] (-7,-2*\yscale-1) to node[midway,below] {time} (7,-2*\yscale-1); 
                \draw[->] (8,-2*\yscale-1) to node[midway,below] {time} (22,-2*\yscale-1); 

            \end{tikzpicture}
            \caption{The worker graph for a task graph with $n=6$, $\gamma=13$, and $\delta=1$.
            Each column of tasks has a separate time axis.
            The tasks within a column are offset such that workers assigned at the same time are located directly above or below each other.}
            \label{fig:worker_graph}
        \end{center}
    \end{figure}
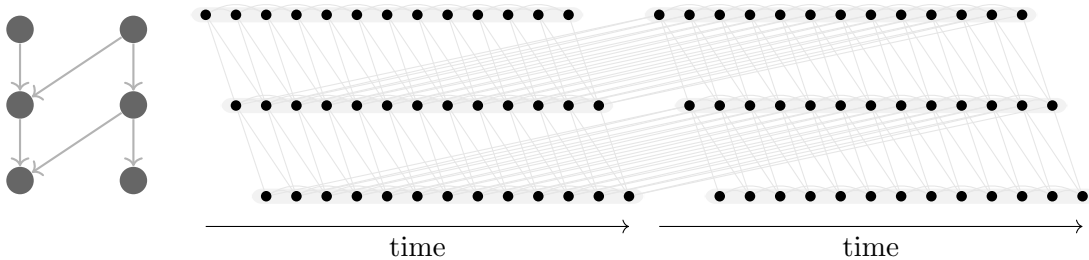

    An execution of the algorithm corresponds to an assignment of workers to the nodes of $G_W$, where each worker is independently malicious with probability at most $\beta$ and otherwise honest.
    We call honest workers who succeed in determining the output of their task \emph{successful} and otherwise \emph{failed}.
    The algorithm is successful when there is a successful worker for each final task.

    \begin{observation}[Successful \& Failed Workers]\label{obs:successful_failed_workers}
    For an assignment of honest and malicious workers to tasks, an honest worker $w = (v,t) \in V_W$ is successful iff for all initial tasks $v'$ that are ancestors of $v$ in $G$, there is a directed path of honest workers connecting an initial worker $w' = (v', t')$ to $w$ in $G_W$.
\end{observation}
    
    We begin our analysis by noting that, from the adversarial workers' perspective, the most effective way to interfere with the computation is simply to avoid ever sending a correct output. 
    Any correct output would only contribute to the computation's success, which contradicts their objective. 
    Hence, to model the strongest possible adversarial behavior, we will w.l.o.g. restrict our analysis to the case where no malicious worker ever sends a correct output.

    \begin{observation}[Dominant Adversarial Strategy]\label{obs:dominant_strategy}
        If the computation succeeds for some assignment of malicious and honest workers, where no malicious worker ever sends a correct output, the computation also succeeds for the same assignment of malicious and honest workers, if there are malicious workers that send some correct outputs.
    \end{observation}

    To prove the correctness of the algorithm, we will proceed as follows.
    First, we show that for every assignment $A$ of honest and malicious workers to the tasks of $G$ that causes the computation to fail, we can identify a certain structure in the worker graph $G_W$, which we call a \emph{valid witness sequence with respect to} $A$.
    We then show that for a random assignment $A$ of workers to tasks, where each worker is independently malicious with probability at most $\beta$ and honest with probability at least $(1-\beta)$, it is unlikely that a valid witness sequence w.r.t. $A$ exists.
    By the contraposition of the first statement, it follows that the computation will succeed with high probability.

    efore we formally define valid witness sequences, we give some intuition.
For simplicity, assume that the task graph is a directed path with $n$ nodes $v_1, \dots, v_n$, where $v_1$ is the initial task and $v_n$ is the final task.
Consider an assignment $A$ of workers to tasks that caused the computation to fail and let $v_f$ be the first task for which no assigned worker is successful.
Intuitively, we want the valid witness sequence $W$ to include a set of malicious workers that is responsible for the failure of task $v_f$. 
Clearly, if all workers assigned to $v_f$ are malicious, we can simply add all of these workers to $W$.
Otherwise, every task $v_1, \dots, v_f$ was assigned some honest worker.
Consider the last task $v$ where the first honest worker assigned to $v$ is successful.
Note that such a task has to exist as all honest workers assigned to the initial task $v_1$ are successful because they receive their input directly from the source.
Let $w_1 = (v,t)$ be the first honest worker assigned to $v$.
We add all malicious workers $M_1 := \{(v,t_{\mathrm{min}}(v)), \dots, (v,t-1)\}$ that were assigned to $v$ before $w_1$ (if there are any), as well as the honest worker $w_1$ to $W$.
Note that by \Cref{obs:successful_failed_workers}, each task above $v$ was assigned a successful worker before time $t$.
We maintain this invariant whenever we add an honest worker to $W$.
Now, for any $i \geq 1$, let $w_i = (v_k, t)$ be the honest worker that was just added to $W$.
We now investigate $w_i$ to decide where the witness sequence continues.

If $w_i$ is successful, the reason for why task $v_f$ failed must lie further down the task graph.
Hence, we continue the witness sequence at the succeeding task $v_{k+1}$, starting with worker $(v_{k+1}, t+1)$.
If all the workers $\{(v_{k+1},t+1), \dots, (v_{k+1}, t_{\mathrm{max}}(v_{k+1}))\}$ are malicious, we add all of them to $W$ and the witness sequence ends (Case (2b) in \Cref{def:witness_sequence}).
Otherwise, let $w_{i+1} = (v_{k+1}, t')$ be the first honest worker assigned to $v_{k+1}$ with $t' > t$. 
Then, we add all the malicious workers $M_{i+1} := \{(v_{k+1},t+1), \dots, (v_{k+1}, t'-1)\}$ (if there are any), as well as $w_{i+1}$, to $W$ and continue the sequence from $w_{i+1}$ (Case (2a) in \Cref{def:witness_sequence}).
Since $w_i$ is successful, we maintain the invariant that each task above $v_{k+1}$ was assigned a successful worker before time $t'$.
Note that it is possible that worker $(v_{k+1}, t+1)$ does not exist if $t+1 < t_{\mathrm{min}}(v_{k+1})$.
In that case, we simply empty $W$ and restart the construction from task $v_{k+1}$.

On the other hand, if $w_i$ failed, there must be a reason for that further up the task graph.
In particular, $w_i$ did not receive a correct output from the preceding task $v_{k-1}$, i.e., none of the workers $\{(v_{k-1},t-2\delta),\dots,(v_{k-1},t-1)\}$ sent a correct output to $w_i$, either because they are malicious or because they failed themselves.
Because of our invariant, we know that there is a successful worker assigned to $v_{k-1}$ before time $t$.
This implies that there must be a contiguous subset $M'_i$ of $2\delta$ malicious workers that was assigned to $v_{k-1}$ after that honest worker but before time $t$, as otherwise one of the workers $\{(v_{k-1},t-2\delta),\dots,(v_{k-1},t-1)\}$ would have sent a correct output to $w_i$.
We add that set $M'_i$ of malicious workers to $W$.
Additionally, and similar to the previous case, we continue the witness sequence at task $v_{k-1}$, starting with worker $(v_{k-1},t)$.
If all the workers $\{(v_{k-1},t), \dots, (v_{k-1}, t_{\mathrm{max}}(v_{k-1}))\}$ are malicious, we add all of them to $W$ and the witness sequence ends (Case (2d) in \Cref{def:witness_sequence}).
Otherwise, let $w_{i+1} = (v_{k-1}, t')$ be the first honest worker assigned to $v_{k-1}$ with $t' \geq t$.
We add all the malicious workers $M_{i+1} := \{(v_{k-1},t), \dots, (v_{k-1}, t'-1)\}$ (if there are any), as well as $w_{i+1}$, to $W$ and continue the sequence from $w_{i+1}$ (Case (2c) in \Cref{def:witness_sequence}). 
Since we move upwards, we clearly maintain the invariant that each task above $v_{k-1}$ was assigned a successful worker before time $t'$.
In case the worker $(v_{k-1},t')$ does not exist because $t > t_{\mathrm{max}}(v_{k-1})$, the sequence ends (also covered by Case (2d) in \Cref{def:witness_sequence}).

Clearly, the witness sequence will never move ``above'' the initial task $v_1$, as all honest workers of $v_1$ are successful, and will never move ``below'' task $v_f$, as all honest workers of $v_f$ failed.
Furthermore, this construction eventually terminates, as the assignment times of the honest workers added to $W$ increase by at least $1$ for each downwards step and never decrease for upwards steps, of which there can be at most $n-1$ in a row.
Hence, whenever an assignment of workers to tasks causes the computation to fail, we can identify a sequence of workers of the above structure.

For general DAGs, we can traverse the task graph in a similar way.
However, for each honest worker added to the witness sequence, there are now several options to which succeeding or preceding task the sequence could move.
If we are not careful in traversing the DAG, the sets of malicious workers $M_i \cup M'_i$ might overlap for different values of $i$, which would decrease the number of workers added to $W$ and, in turn, increase the probability that the valid witness sequence exists.

In order to analyze the probability of the computation failing, we now formally define witness sequences.
Intuitively, a witness sequence is any sequence of workers $W$ that in principle \emph{could} be created by the construction described above for some assignments of workers to tasks, i.e., the definition is independent of which workers are actually malicious and which are honest.
This allows us to count the total number of witness sequences.
An example of a witness sequence is depicted in \Cref{fig:witness_sequence}.

    \begin{definition}[Witness Sequence]\label{def:witness_sequence}
        Let $G=(V,E)$ be a task graph, and $G_W = (V_W, E_W)$ be the worker graph of $G$.
        We call a sequence $W=((M_1,M_1',w_1),\dots,(M_\ell,M_\ell',w_\ell),(M_{\ell+1},M_{\ell+1}'=\emptyset,w_{\ell+1}=\perp))$ with $\ell\geq 0$ a witness sequence, if $w_i \in V_W$ for all $1\leq i\leq \ell$, each $M_i \subseteq V_W$ and $M'_i \subseteq V_W$ with $1\leq i\leq \ell+1$ is a contiguous subsets of workers assigned to the same task, and all of the following properties hold:
        \begin{enumerate}
            \item $(v,t_\mathrm{min}(v))\in M_1$ for some $v\in V$.
            \item For $(M_i,M'_i,w_i)$ and $(M_{i+1},M'_{i+1},w_{i+1})$ in $W$ with $1\leq i\leq \ell-1$ and $w_i=(v,t)$, exactly one of the following holds:
            \begin{enumerate}
                \item \textbf{Downwards:} $w_{i+1} = (v',t')$ with $v'\in\suc_G(v)$ and $t'>t$, and $M_{i+1}=\{(v',t+1),\dots,(v',t'-1)\}$ (possibly empty). Further, $M_{i}' = \emptyset$.
                \item \textbf{Downwards \& Stop:} $w_{i+1} = \perp$ and $M_{i+1}=\{(v',t+1),\allowbreak\dots,\allowbreak(v',t_\mathrm{max}(v'))\}$ with $v' \in \suc_G(v)$. Further $M'_{i}=\emptyset$.
                \item \textbf{Upwards:} $w_{i+1} = (v',t')$ with $v'\in\pre_G(v)$ and $t'\geq t$, and $M_{i+1}=\{(v',t),\dots,(v',t'-1)\}$ (possibly empty), $t \geq t_\mathrm{min}(v) + \delta$, and further $M'_{i}\coloneqq\{(v',t''),\allowbreak\dots,\allowbreak(v',\min\{t''+2\delta-1,t_\mathrm{max}(v')\})\}$ for a $t'' \leq t - 2\delta$.
                \item \textbf{Upwards \& Stop:} $w_{i+1} = \perp$ and $M_{i+1}=\{(v',t),\allowbreak\dots,\allowbreak(v',t_\mathrm{max}(v'))\}$ with $v' \in \pre_G(v)$ (possibly empty). Further, $t \geq t_\mathrm{min}(v) + \delta$, and $M'_{i}\coloneqq\{(v',t''),\dots,(v',\min\{t''+2\delta-1,t_\mathrm{max}(v')\})\}$ for a $t'' \leq t - 2\delta$.
            \end{enumerate}
            \item The sets $M_1,\dots,M_{\ell+1},M'_1,\dots,M'_{\ell+1}$ are pairwise disjoint.
            \item Exactly one of the following holds:
            \begin{itemize}
                \item $(v, t_\mathrm{max}(v)) \in M_{\ell + 1}$ for some $v \in V$.
                \item $\ell \geq 1$, $w_\ell \in \{(v, t_{\mathrm{max}}(v) - \delta + 1), \dots, (v, t_{\mathrm{max}}(v))\}$ for some $v \in V$ and $M_{\ell + 1} = \emptyset$.
            \end{itemize}
        \end{enumerate}
        Further, we denote $\mathcal{M}_W\coloneqq M_1\cup\dots\cup M_{\ell+1}$, $\mathcal{M}'_W\coloneqq M'_1\cup\dots\cup M'_{\ell+1}$,  $\mathcal{H}_W\coloneqq\{w_i \mid 1\leq i\leq \ell\}$, and $\mathcal{F}_W\coloneqq\{w_i=(v,t) \mid 1\leq i\leq \ell \; \land \; M_{i+1}\cup\{w_{i+1}\} \text{ contains no $(v',t')$ with $v'\in\suc_G(v)$}\}$.
        If the witness sequence we refer to is obvious from the context, we omit the subscripts.
    \end{definition}

    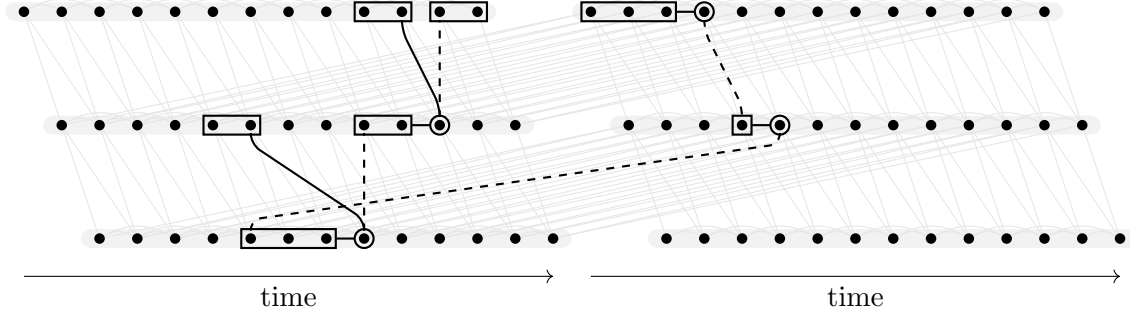
\begin{figure}[ht!]
        \begin{center}
            \begin{tikzpicture}[scale=0.5]
                \def\layers
                {
                    {{0,0,0,0,0,0,0,0,1,0,0,0,0},{0,0,0,1,0,1,0,0,0,0,0,0,0}},
                    {{0,0,0,1,0,0,0,0,0,0,2,0,0},{0,0,0,0,1,0,0,1,0,0,0,0,0}},
                    {{0,0,0,0,0,0,0,2,0,0,0,0,0},{0,0,0,0,0,0,0,0,0,0,0,0,0}},
                }
                \def\numoffset{1}
                \def\numsuccs{2}
                \def\numassign{13}

                \def\xoffset{2}
                \def\yscale{3}

                \foreach \layer [count=\y from 0] in \layers {
                    \pgfmathparse{dim({\layer})}
                    \let\layersize\pgfmathresult
                    \pgfmathparse{-(\layersize-1)/2}
                    \let\xstart\pgfmathresult
                    \foreach \task [count=\xx from 0] in \layer{
                        \pgfmathparse{\y+(\numassign+\xoffset)*(\xstart+\xx)}
                        \let\x\pgfmathresult
                        \fill[black!5,rounded corners] (\x,-\y*\yscale-0.25) rectangle (\x+\numassign,-\y*\yscale+0.25);
                        \foreach \worker [count=\w] in \task {
                            \node[circle,fill=black,inner sep=1.4pt] (\y-\xx-\w) at (\x+\w-0.5,-\y*\yscale) {};
                        }
                    }
                }

                \foreach \y/\x/\xx in {0/0/0,0/1/0,0/1/1,1/0/0,1/1/0,1/1/1} {
                    \pgfmathparse{int(\y+1)}
                    \let\yy\pgfmathresult
                    \foreach \w in {1,...,\numassign}{
                        \pgfmathparse{int(\w+1-\numoffset)}
                        \let\wmin\pgfmathresult
                        \pgfmathparse{int(\w+\numsuccs-\numoffset)}
                        \let\wmax\pgfmathresult
                        \foreach \ww in {\wmin,...,\wmax}{
                            \ifnum\ww>\numassign\else
                                \draw[black!10] (\y-\x-\w) -- (\yy-\xx-\ww);
                            \fi
                        }
                    }
                }

                    \foreach \layer [count=\y from 0] in \layers {
                    \foreach \task [count=\x from 0] in \layer{
                        \foreach \worker [count=\w] in \task {
                            \pgfmathparse{int(\w+1)}
                            \let\wmin\pgfmathresult
                            \pgfmathparse{int(\w+\numsuccs)}
                            \let\wmax\pgfmathresult
                            \foreach \ww in {\wmin,...,\wmax}{
                                \ifnum\ww>\numassign\else
                                    \draw[black!10] (\y-\x-\w) to[bend left] (\y-\x-\ww);
                                \fi
                            }
                        }
                    }
                }

                \foreach \y/\x/\w/\xx in {0/1/4/1,1/1/5/0} {
                    \pgfmathparse{int(\y+1)}
                    \let\yy\pgfmathresult
                    \draw[thick,rounded corners,dashed] ($(\y-\x-\w)+(0,-0.25)$) -- ($(\y-\x-\w)+(0,-0.5)$) -- ($(\yy-\xx-\w)+(0,+0.5)$) -- ($(\yy-\xx-\w)+(0,+0.25)$);
                }
                \foreach \y/\x/\w/\xx in {2/0/8/0,1/0/11/0} {
                    \pgfmathparse{int(\y-1)}
                    \let\yy\pgfmathresult
                    \pgfmathparse{int(\w+1)}
                    \let\ww\pgfmathresult
                    \draw[thick,rounded corners,dashed] ($(\y-\x-\w)+(0,+0.25)$) -- ($(\y-\x-\w)+(0,+0.5)$) -- ($(\yy-\xx-\ww)+(0,-0.5)$) -- ($(\yy-\xx-\ww)+(0,-0.25)$);
                }
                
                \foreach \y/\x/\w/\xx/\ww in {2/0/8/0/6,1/0/11/0/11} {
                    \pgfmathparse{int(\y-1)}
                    \let\yy\pgfmathresult
                    \draw[thick,rounded corners] ($(\y-\x-\w)+(0,+0.25)$) -- ($(\y-\x-\w)+(0,+0.5)$) -- ($(\yy-\xx-\ww)+(0,-0.5)$) -- ($(\yy-\xx-\ww)+(0,-0.25)$);
                }
                \foreach \y/\x/\w/\ww in {1/0/5/6,0/0/10/11} {
                    \draw[thick] ($(\y-\x-\w)+(-0.25,-0.25)$) rectangle ($(\y-\x-\ww)+(+0.25,+0.25)$);
                }
                
                \foreach \y/\x/\w in {0/1/4,1/1/5} {
                    \draw[thick] ($(\y-\x-\w)+(-0.25,0)$) -- ($(\y-\x-\w)+(-0.75,0)$);
                }
                \foreach \y/\x/\w in {0/1/4,1/1/5} {
                    \draw[thick] (\y-\x-\w) circle (0.25);
                }
                
                \foreach \y/\x/\w in {2/0/8,1/0/11} {
                    \draw[thick] ($(\y-\x-\w)+(-0.25,0)$) -- ($(\y-\x-\w)+(-0.75,0)$);
                }
                \foreach \y/\x/\w in {2/0/8,1/0/11} {
                    \draw[thick] (\y-\x-\w) circle (0.25);
                }
                
                \foreach \y/\x/\w/\ww in {0/1/1/3,1/1/4/4,2/0/5/7,1/0/9/10,0/0/12/13} {
                    \draw[thick] ($(\y-\x-\w)+(-0.25,-0.25)$) rectangle ($(\y-\x-\ww)+(+0.25,+0.25)$);
                }

                \draw[->] (-7,-2*\yscale-1) to node[midway,below] {time} (7,-2*\yscale-1); 
                \draw[->] (8,-2*\yscale-1) to node[midway,below] {time} (22,-2*\yscale-1); 

            \end{tikzpicture}
            \caption{An example of a witness sequence starting from the top right task and ending at the top left task. Consecutive triples in the witness sequence are connected by a dashed line. For each triple $(M,M',w)$, $w$ is enclosed by a circle, the workers of $M$ are enclosed by a box adjacent to $w$, and the workers of $M'$ are enclosed by a box at the next task visited by the witness sequence. $w$ is connected to both $M$ and $M'$ by a solid line. $w_{\ell+1}=\perp$ is not depicted.}
            \label{fig:witness_sequence}
        \end{center}
    \end{figure}

    We now define what it means for a witness sequence to be \emph{valid} with respect to an assignment of workers to tasks.
    For a witness sequence $W$ to be valid, we require that for each triple $(M_i, M'_i, w_i)$ in $W$, all workers in $M_i$ and $M'_i$ are malicious and that $w_i$ is honest.
    Moreover, if $W$ continues upwards, $w_i$ must be a failed worker, whereas if $W$ continues downwards, $w_i$ must be a successful worker.
    An example of a witness sequence that is valid w.r.t an assignment is depicted in \Cref{fig:valid_witness_sequence}.

    \begin{definition}[Validity w.r.t.\ an Assignment]\label{def:valid_witness_sequence}
        For a task graph $G=(V,E)$ and its worker graph $G_W=(V_W,E_W)$, let $A$ be an assignment of workers to tasks.
        We call a witness sequence $W=((M_1,M'_1,w_1),\dots,(M_\ell,M'_\ell,w_\ell),(M_{\ell+1},\emptyset,\perp))$ valid w.r.t. $A$, if all workers in $\mathcal{M} \cup \mathcal{M}'$ are malicious, all workers in $\mathcal{F}$ are honest workers that failed and all workers in $\mathcal{H} \setminus \mathcal{F}$ are honest workers that were successful.
    \end{definition}

    \begin{figure}[ht!]
        \begin{center}
            \begin{tikzpicture}[scale=0.5]
                \def\layers
                {
                    {{0,1,1,1,0,0,1,0,1,0,0,0,0},{0,0,0,1,0,1,0,1,0,0,1,1,0}},
                    {{0,1,0,1,0,0,0,0,0,0,2,0,0},{2,0,2,0,1,0,0,1,0,1,0,0,0}},
                        {{2,0,0,0,0,0,0,2,0,2,2,0,2},{0,2,0,2,1,0,0,0,1,0,0,2,2}},
                }
                \def\numoffset{1}
                \def\numsuccs{2}
                \def\numassign{13}

                \def\xoffset{2}
                \def\yscale{3}

                \foreach \layer [count=\y from 0] in \layers {
                    \pgfmathparse{dim({\layer})}
                    \let\layersize\pgfmathresult
                    \pgfmathparse{-(\layersize-1)/2}
                    \let\xstart\pgfmathresult
                    \foreach \task [count=\xx from 0] in \layer{
                        \pgfmathparse{\y+(\numassign+\xoffset)*(\xstart+\xx)}
                        \let\x\pgfmathresult
                        \fill[black!5,rounded corners] (\x,-\y*\yscale-0.25) rectangle (\x+\numassign,-\y*\yscale+0.25);
                        \foreach \worker [count=\w] in \task {
                            \ifnum\worker=0
                                \node[circle,fill=red,inner sep=1.4pt] (\y-\xx-\w) at (\x+\w-0.5,-\y*\yscale) {};
                            \fi
                            \ifnum\worker=1
                                \node[circle,fill=green!70!black,inner sep=1.4pt] (\y-\xx-\w) at (\x+\w-0.5,-\y*\yscale) {};
                            \fi
                            \ifnum\worker=2
                                \node[circle,fill=yellow!90!black,inner sep=1.4pt] (\y-\xx-\w) at (\x+\w-0.5,-\y*\yscale) {};
                            \fi
                        }
                    }
                }

                \foreach \y/\x/\xx in {0/0/0,0/1/0,0/1/1,1/0/0,1/1/0,1/1/1} {
                    \pgfmathparse{int(\y+1)}
                    \let\yy\pgfmathresult
                    \foreach \w in {1,...,\numassign}{
                        \pgfmathparse{int(\w+1-\numoffset)}
                        \let\wmin\pgfmathresult
                        \pgfmathparse{int(\w+\numsuccs-\numoffset)}
                        \let\wmax\pgfmathresult
                        \foreach \ww in {\wmin,...,\wmax}{
                            \ifnum\ww>\numassign\else
                                \draw[black!10] (\y-\x-\w) -- (\yy-\xx-\ww);
                            \fi
                        }
                    }
                }

                \foreach \layer [count=\y from 0] in \layers {
                    \foreach \task [count=\x from 0] in \layer{
                        \foreach \worker [count=\w] in \task {
                            \pgfmathparse{int(\w+1)}
                            \let\wmin\pgfmathresult
                            \pgfmathparse{int(\w+\numsuccs)}
                            \let\wmax\pgfmathresult
                            \foreach \ww in {\wmin,...,\wmax}{
                                \ifnum\ww>\numassign\else 
                                    \draw[black!10] (\y-\x-\w) to[bend left] (\y-\x-\ww);
                                \fi
                            }
                        }
                    }
                }

                \foreach \y/\x/\w/\xx in {0/1/4/1,1/1/5/0} {
                    \pgfmathparse{int(\y+1)}
                    \let\yy\pgfmathresult
                    \draw[thick,rounded corners,dashed] ($(\y-\x-\w)+(0,-0.25)$) -- ($(\y-\x-\w)+(0,-0.5)$) -- ($(\yy-\xx-\w)+(0,+0.5)$) -- ($(\yy-\xx-\w)+(0,+0.25)$);
                }
                
                \foreach \y/\x/\w/\xx in {2/0/8/0,1/0/11/0} {
                    \pgfmathparse{int(\y-1)}
                    \let\yy\pgfmathresult
                    \pgfmathparse{int(\w+1)}
                    \let\ww\pgfmathresult
                    \draw[thick,rounded corners,dashed] ($(\y-\x-\w)+(0,+0.25)$) -- ($(\y-\x-\w)+(0,+0.5)$) -- ($(\yy-\xx-\ww)+(0,-0.5)$) -- ($(\yy-\xx-\ww)+(0,-0.25)$);
                }
                
                \foreach \y/\x/\w/\xx/\ww in {2/0/8/0/6,1/0/11/0/11} {
                    \pgfmathparse{int(\y-1)}
                    \let\yy\pgfmathresult
                    \draw[thick,rounded corners] ($(\y-\x-\w)+(0,+0.25)$) -- ($(\y-\x-\w)+(0,+0.5)$) -- ($(\yy-\xx-\ww)+(0,-0.5)$) -- ($(\yy-\xx-\ww)+(0,-0.25)$);
                }
                \foreach \y/\x/\w/\ww in {1/0/5/6,0/0/10/11} {
                    \draw[red, thick] ($(\y-\x-\w)+(-0.25,-0.25)$) rectangle ($(\y-\x-\ww)+(+0.25,+0.25)$);
                }
                
                \foreach \y/\x/\w in {0/1/4,1/1/5} {
                    \draw[thick] ($(\y-\x-\w)+(-0.25,0)$) -- ($(\y-\x-\w)+(-0.75,0)$);
                }
                \foreach \y/\x/\w in {0/1/4,1/1/5} {
                    \draw[green!70!black, thick] (\y-\x-\w) circle (0.25);
                }
                
                \foreach \y/\x/\w in {2/0/8,1/0/11} {
                    \draw[thick] ($(\y-\x-\w)+(-0.25,0)$) -- ($(\y-\x-\w)+(-0.75,0)$);
                }
                \foreach \y/\x/\w in {2/0/8,1/0/11} {
                    \draw[yellow!90!black, thick] (\y-\x-\w) circle (0.25);
                }
                
                \foreach \y/\x/\w/\ww in {0/1/1/3,1/1/4/4,2/0/5/7,1/0/9/10,0/0/12/13} {
                    \draw[red, thick] ($(\y-\x-\w)+(-0.25,-0.25)$) rectangle ($(\y-\x-\ww)+(+0.25,+0.25)$);
                }

                \draw[->] (-7,-2*\yscale-1) to node[midway,below] {time} (7,-2*\yscale-1); 
                \draw[->] (8,-2*\yscale-1) to node[midway,below] {time} (22,-2*\yscale-1); 

            \end{tikzpicture}
            \caption{A witness sequence that is valid w.r.t. to an assignment $A$.
            To depict $A$, successful workers are presented in green color, failed workers are presented in yellow, and malicious workers are presented in red.
            For each triple $(M,M',w)$, $M$ and $M'$ contain only malicious workers, and $w$ is an honest worker that is either successful if the sequence continues at a succeeding task or failed if the sequence continues at a preceding task.}
            \label{fig:valid_witness_sequence}
        \end{center}
    \end{figure}
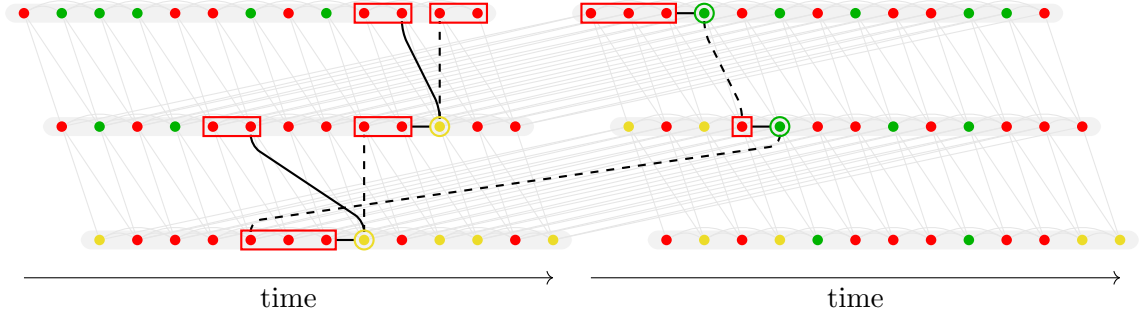

    With these definitions in place, we are ready to prove the correctness of the algorithm.
    The first step is to prove that a failing execution of the algorithm implies the existence of a valid witness sequence.
    In \Cref{sec:valid_witness_sequence_existence_paths}, we provide intuition for the more general statement for DAGs by proving the following lemma stating that a valid witness sequence must exist when the computation fails, if the task graph is a path graph.
    \begin{lemma}[Valid Witness Sequence Existence]\label{lem:valid_witness_existence}
        For a path graph $G=(V,E)$ with $V=\{v_1,\dots,v_n\}$ and its worker graph $G_W=(V_W,E_W)$, we consider an assignment $A$ of workers to tasks which caused the computation to fail.
        Then, the construction in \Cref{def:witness_sequence_construction} terminates and constructs a witness sequence that is valid w.r.t. $A$.
    \end{lemma}
    Afterwards, in \Cref{sec:valid_witness_sequence_existence_dags}, we prove the statement for any DAG formulated in the following lemma.
            \begin{lemma}[Valid Witness Sequence Existence]\label{lem:valid_witness_existence_dags}
        For a task graph $G=(V,E)$ and its worker graph $G_W=(V_W,E_W)$, we consider an assignment $A$ of workers to tasks which caused the computation to fail.
        Then, there is a witness sequence in $G_W$ that is valid w.r.t. $A$.
    \end{lemma}
    Finally, in \Cref{sec:correctness_and_work_bound} we conclude the analysis of the correctness by proving that it is unlikely for a valid witness sequence to exist.
    Additionally, we provide a bound on the work performed by the workers and prove the main theorem of this paper stated in the following. 
    \begin{theorem}\label{th:main_theorem_dags} 
        Let $G$ be a task graph and $\beta <1$ be any fixed value representing the fraction of adversarial workers. 
        For $\delta = \Theta(\log_{\nicefrac{1}{\beta}} (d) \cdot \log_{\nicefrac{1}{\beta}} \log n)$ and $\gamma = \Theta(\log n)$, the target learns the correct solution w.h.p., and the algorithm terminates within $O(D\log_{\nicefrac{1}{\beta}} (d) \cdot \log_{\nicefrac{1}{\beta}}\log n+\log n)$ rounds. 
        Each honest worker performs at most one computation for each time it was assigned, $O(d\log_{\nicefrac{1}{\beta}} (d) \cdot \log_{\nicefrac{1}{\beta}}\log n)$ verifications, and $O(d\log_{\nicefrac{1}{\beta}} (d) \cdot \log_{\nicefrac{1}{\beta}}\log n)$ communication. In expectation, each task gets executed by $1+o(1)$ honest workers and the total number of computations performed by honest workers is $n(1+o(1))$. The source needs to send the input $O(\log n)$ times for each initial task of $G$, the target needs to receive and verify the output $O(\log n)$ times for each final task of $G$, and the supervisor performs $O(n \log n)$ assignments and $O(nd \log (n) \cdot \log_{\nicefrac{1}{\beta}} (d) \cdot \log_{\nicefrac{1}{\beta}}\log n)$ introductions.
    \end{theorem}
    
    \subsection{Witness Existence for the Path Case}
    \label{sec:valid_witness_sequence_existence_paths}

    Let $G=(V,E)$ be a directed path with $n$ nodes $v_1, \dots, v_n$, where $v_1$ is the initial task and $v_n$ is the final task, and let $G_W = (V_W, E_W)$ be the worker graph of $G$.
    Before we show how to construct a valid witness sequence in $G_W$, we prove three simple properties that we require for the construction to be well-defined.
    
    \begin{lemma}[Properties of Assignments]\label{lem:assignment_properties}
        Let $G=(V,E)$ be a path graph with $V=\{v_1,\dots,v_n\}$ and let $G_W=(V_W,E_W)$ be its worker graph.
        For an assignment $A$ of workers to tasks, the following holds:
        \begin{itemize}
            \item If there is no task for which all assigned workers are malicious, there is a task $v$ where the first honest worker assigned to $v$ is successful.
            \item If a worker $(v_k, t)$ failed, the task $v_{k-1}$ exists.
            \item If $A$ caused the computation to fail and a worker $(v_k, t)$ is successful, the task $v_{k+1}$ exists.
        \end{itemize}
    \end{lemma}
    \begin{proof}
        If there is no task for which all assigned workers are malicious, task $v_1$ is assigned at least one honest worker.
        Since $v_1$ is an initial task, all honest workers assigned to $v_1$ receive their input from the source and are thus successful.
        In particular, the first honest worker assigned to $v_1$ is successful.
        Therefore, there is a task $v$, where the first honest worker assigned to $v$ is successful. 
        Next, consider a worker $(v_k, t)$ that failed.
        Since all honest workers of task $v_1$ are successful, we have $1 < k \leq n$.
        Hence, task $v_{k-1}$ exists.
        Finally, if $A$ caused the computation to fail, there is a task for which no assigned worker is successful as the target would have otherwise received the output.
        Let $v_j$ be the first such task, i.e., no task in $\{v_j, \dots, v_n\}$ is assigned a worker that is successful.
        Hence, for all successful workers $(v_k, t)$, we have $1 \leq k < j \leq n$ and the task $v_{k+1}$ exists.
    \end{proof}

    \begin{definition}[Valid Witness Sequence Construction for Path Graphs]\label{def:witness_sequence_construction}
        For a path graph $G=(V,E)$ with $V=\{v_1,\dots,v_n\}$ and its worker graph $G_W=(V_W,E_W)$, we consider an assignment of workers to tasks which caused the computation to fail.
        We construct a valid witness sequence $W$ by iteratively adding triples $(M,M',w)\in\mathcal{P}(V_W)\times \mathcal{P}(V_W) \times V_W$ to $W$ until a triple $(M,\emptyset,\perp)$ is added that marks the end of the construction.
    
        If there is a task $v \in V$ such that all workers assigned to $v$ are malicious, we add the triple $(M_1, \emptyset, \perp)$ to $W$ with $M_1 = \{(v, t_{\mathrm{min}}(v)), \dots, (v, t_{\mathrm{max}}(v))\}$.
        Otherwise, we consider the last task $v$, where the first honest worker $(v,t)$ assigned to $v$ is successful (which exists due to \Cref{lem:assignment_properties}), and add $(M_1, \emptyset, (v,t))$ to $W$ with $M_1 = \{(v,t_\mathrm{min}(v)),\dots,(v,t-1)\}$. 
        After adding $(M_i,M_i',w_i)$ to $W$, if $w_i=\perp$ our construction terminates. Otherwise, if $w_i=(v_k,t)\in V_W$ we add $(M_{i+1}, M'_{i+1}, w_{i+1})$ to $W$. To define $M_{i+1}$ and $w_{i+1}$, we distinguish two cases:
        \newpage
        \begin{enumerate}
            \item $w_i$ is successful. 
            Then $v_{k+1}$ exists due to \Cref{lem:assignment_properties}.
            \begin{enumerate}
                \item If $t+1<t_\mathrm{min}(v_{k+1})$, we empty $W$ and restart the construction (setting $i\coloneqq 0$).
                We set $w_1=(v_{k+1},t')$ to be the first honest worker assigned to task $v_{k+1}$ and $M_1 = \{(v_{k+1},t_\mathrm{min}(v_{k+1})),\dots,(v_{k+1},t'-1)\}$ (possibly empty).
                \item Otherwise, if there is no honest worker in $M\coloneqq\{(v_{k+1},t+1),\dots,(v_{k+1},t_\mathrm{max}(v_{k+1}))\}$, we set $M_{i+1} = M$ and $w_{i+1} = \perp$. 
                \item Otherwise, we set $w_{i+1}=(v_{k+1},t')$ to be the first honest worker in $M$ (of (b)) and $M_{i+1} = \{(v_{k+1},t+1),\dots,(v_{k+1},t'-1)\}$ (possibly empty). 
            \end{enumerate} 
            \item $w_i$ failed. Then $v_{k-1}$ exists due to \Cref{lem:assignment_properties}.
            \begin{enumerate}
                \item If $t>t_\mathrm{max}(v_{k-1})$, we set $M_{i+1}=\emptyset$  and $w_{i+1}=\perp$.
    
                \item Otherwise, if there is no honest worker in $M\coloneqq\{(v_{k-1},t),\dots,(v_{k-1},t_\mathrm{max}(v_{k-1}))\}$, we set $M_{i+1}=M$ and $w_{i+1}=\perp$.
    
                \item Otherwise, we set $w_{i+1}=(v_{k-1},t')$ to be the first honest worker in $M$ (of (b)) and $M_{i+1} = \{(v_{k-1},t),\dots,(v_{k-1},t'-1)\}$ (possibly empty). 
            \end{enumerate}
        \end{enumerate}
        Additionally, if $w_{i+1}=\perp$ or $w_{i+1}$ is successful, we define $M'_{i+1}=\emptyset$, and if $w_{i+1}=(v_{k'},t')$ failed, we define $M'_{i+1} = \{(v_{k'-1},t''), \dots, (v', \min\{t''+2\delta-1,t_\mathrm{max}(v_{k'-1})\})\}$, where $(v_{k'-1},t''-1)$ is the last successful worker assigned to $v_{k'-1}$ with $t''-1<t'$ (which exists due to \Cref{lem:success_barrier_existence}).
    \end{definition}
    
    Note that whenever an honest worker $w_{i+1}$ is added in the construction above, it was assigned after or at the same time $w_i$ was assigned.
    We formulate this fact in the following observation.
    
    \begin{observation}[Non-decreasing $w_i$]\label{obs:non-decreasing}
        For a task graph $G=(V,E)$ and its worker graph $G_W=(V_W,E_W)$, we consider an assignment of workers to tasks which caused the computation to fail.
        Let $W=((M_1,M'_1,w_1),(M_2,M'_2,w_2),\dots)$ be a sequence constructed by \Cref{def:witness_sequence_construction}.
        For each $1\leq i<j\leq \ell$ with $w_i = (v, t_i)$ and $w_j = (v',t_j)$, we have $t_i\leq t_j$.
        Additionally, for each worker $w_i=(v,t_i)$, no worker in $M_i$ was assigned before time $t_i$ and all workers in $M'_i$ were assigned before time $t_i$.
    \end{observation}

    For the construction in \Cref{def:witness_sequence_construction} to be well-defined, we still need to argue that whenever we move upwards from a failed worker $w_i = (v_k, t)$, there is indeed a successful worker assigned to task $v_{k-1}$ before time $t$.
    In fact, we can show something even stronger.
    Whenever we add an honest worker $w_i = (v_k, t)$ to $W$, there is a successful worker that was assigned before time $t$ for \emph{every} task $v_j$ with $1 \leq j < k$.
    Moreover, if $w''_j = (v_j, t'' - 1)$ is the last successful worker assigned to task $v_j$ before time $t$, there is no failed worker $(v_{j+1}, t')$ in $W$ with $t'' - 1 < t' < t$.
    In other words, no worker $w''_j$ has been ``used'' before by some other failed worker $w_{i'}$ in $W$ to define its own set $M'_{i'}$ for any $i' < i$.
    We call $\{w''_1,\dots,w''_{k-1}\}$ the \emph{success barrier} for $w_i$.
    An example of a success barrier is depicted in \Cref{fig:success_barrier}.
    The existence of a success barrier will be crucial later to argue that all of the sets $M'_i$ created by the construction are disjoint.
    We now formally define the success barrier and prove that every honest worker added to $W$ indeed has one.
    
    \begin{definition}[Success Barrier]\label{def:success_barrier}
        For a path graph $G=(V,E)$ with $V=\{v_1,\dots,v_n\}$ and its worker graph $G_W=(V_W,E_W)$, we consider an assignment of workers to tasks which caused the computation to fail.
        Let $W$ be the sequence constructed by \Cref{def:witness_sequence_construction}.
        We call a successful worker $w'' = (v_k, t''-1)$ \emph{linked at time $t$} if there is a failed worker $w' = (v_{k+1}, t')$ in $W$ such that $t''-1 < t' < t$. 
    
        Let $(M,M',w)$ be a triple we add to $W$ (cf. \Cref{def:witness_sequence_construction}) with $w=(v_k,t)$. 
        If for all $1\leq j< k$, there is a successful worker assigned to task $v_j$ before time $t$ and the last such worker $w''_j$ is not linked at time $t$, we call $\{w''_1,\dots,w''_{k-1}\}$ the success barrier for $w$.
    \end{definition}
    
    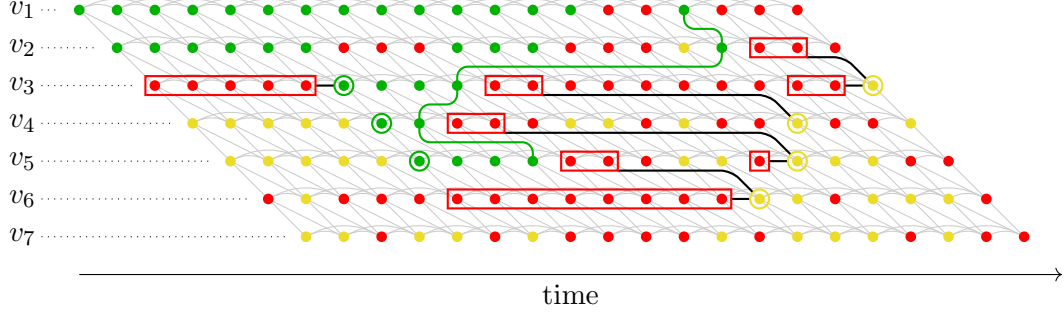
\begin{figure}[ht!]
        \begin{center}
            \begin{tikzpicture}[scale=0.5]
                \def\assignment
                {
                    {1,1,1,1,1,1,1,1,1,1,1,1,1,1,0,0,1,0,0,0},
                    {1,1,1,1,1,1,0,0,0,1,1,1,0,0,0,2,1,0,0,0},
                        {0,0,0,0,0,1,1,1,1,0,0,0,0,0,0,0,0,0,0,2},
                        {2,2,2,2,2,1,1,0,0,0,2,2,0,2,0,0,2,0,0,2},
                            {2,2,2,2,2,1,1,1,1,0,0,0,2,2,0,2,2,2,0,0},
                            {0,2,0,0,0,0,0,0,0,0,0,0,0,2,2,0,2,2,2,0},
                                {2,2,0,2,2,0,2,0,0,0,0,2,0,2,2,2,0,2,0,0},
                }
                \def\numoffset{1}
                \def\numsuccs{2}
                \def\numassign{20}
                \def\numtasks{7}

                \foreach \task [count=\y from 0] in \assignment {
                    \foreach \worker [count=\x from 0] in \task {
                        \foreach \xoffset in {1,...,\numsuccs} {
                            \pgfmathparse{\x+\xoffset < \numassign}
                            \ifnum\pgfmathresult=1
                                \draw[black!20] (\x+\y*\numoffset,-\y) to[bend left] (\x+\xoffset+\y*\numoffset,-\y-0);
                            \fi
                            \pgfmathparse{\x+\xoffset >= \numoffset && \x+\xoffset < \numassign+\numoffset && \y < \numtasks-1}
                            \ifnum\pgfmathresult=1
                                \draw[black!20] (\x+\y*\numoffset,-\y) to (\x+\xoffset+\y*\numoffset,-\y-1);
                            \fi
                        }
                    }
                }

                \foreach \y/\x in {2/7} {
                    \draw[thick] (\x-0.25,-\y) -- (\x-0.75,-\y);
                }
                \foreach \y/\x in {2/7,3/8,4/9} {
                    
                    \draw[green!70!black, thick] (\x,-\y) circle (0.25);
                    
                }
                
                \foreach \y/\x in {5/18,4/19,2/21} {
                    \draw[thick] (\x-0.25,-\y) -- (\x-0.75,-\y);
                }
                \foreach \y/\x in {5/18,4/19,3/19,2/21} {
                    
                    \draw[yellow!90!black, thick] (\x,-\y) circle (0.25);
                }
                
                \foreach \y/\xstart/\xend in {2/2/6,5/10/17,4/18/18,2/19/20} {
                    \draw[red, thick] (\xstart-0.25,-\y-0.25) rectangle (\xend+0.25,-\y+0.25);
                }
                
                \def\xx{12}
                \def\yy{4}
                \foreach \y/\x in {3/9,2/10,1/17,0/16} {
                    \draw[green!70!black,thick,rounded corners] (\xx,-\yy) to (\xx,-\yy+0.5) to (\x,-\y-0.5) to (\x,-\y);
                    \xdef\xx{\x}
                    \xdef\yy{\y}
                }
                
                \foreach \y/\x/\yy/\xx in {5/18/4/14,4/19/3/11,3/19/2/12,2/21/1/19} {
                    \draw[thick,rounded corners] (\x-0.19,-\y+0.19) to (\x-0.75,-\y+0.75) to (\xx+0.25,-\yy-0.25);
                }
                \foreach \y/\xstart/\xend in {4/13/14,3/10/11,2/11/12,1/18/19} {
                    \draw[red, thick] (\xstart-0.25,-\y-0.25) rectangle (\xend+0.25,-\y+0.25);
                }

                \foreach \task [count=\y from 0] in \assignment {
                    \foreach \worker [count=\x from 0] in \task {
                        \ifthenelse{\worker = 0}
                        {
                            \fill[red] (\x+\y*\numoffset,-\y) circle (4pt);
                        }{}
                        \ifthenelse{\worker = 1}
                        {
                            \fill[green!70!black] (\x+\y*\numoffset,-\y) circle (4pt);
                        }{}
                        \ifthenelse{\worker = 2}
                        {
                            \fill[yellow!90!black] (\x+\y*\numoffset,-\y) circle (4pt);
                        }{}
                        \ifthenelse{\worker = 3}
                        {
                            \fill[red] (\x+\y*\numoffset,-\y) circle (4pt);
                            \fill[white] (\x+\y*\numoffset,-\y) circle (1pt);
                        }{}
                        \ifthenelse{\worker = 4}
                        {
                            \fill[blue!90!black] (\x+\y*\numoffset,-\y) circle (4pt);
                            \fill[white] (\x+\y*\numoffset,-\y) circle (1pt);
                        }{}
                    }
                }

                \draw[->] (0,-\numtasks) to node[midway,below] {time} (\numassign+\numoffset*\numtasks-\numoffset,-\numtasks); 
                \foreach \y in {1,...,\numtasks} {
                    \node (v\y) at (-1.5,-\y+1) {$v_\y$};
                    \draw[dotted] (-1,-\y+1) -- (\y*\numoffset-1.5,-\y+1);
                }
                
            \end{tikzpicture}
            \caption{An example of the success barrier for a path graph with $7$ nodes, where $\gamma = 20$ and $\delta = 1$. The green line connects the nodes in the success barrier for the first failed worker in the witness sequence.
            We can link each failed worker with a unique honest node in its success barrier and assigned to the preceding task.
            This honest worker must be succeeded by $2\delta=2$ malicious workers that are not in the witness sequence, which compose the set $M'$ corresponding to that failed worker.}
            \label{fig:success_barrier}
        \end{center}
    \end{figure}    
    
    \begin{lemma}[Success Barrier Existence]\label{lem:success_barrier_existence}
        For a path graph $G=(V,E)$ with $V=\{v_1,\dots,v_n\}$ and its worker graph $G_W=(V_W,E_W)$, we consider an assignment of workers to tasks which caused the computation to fail.
        Let $W=((M_1,M'_1,w_1), (M_2,M'_2,w_2), \dots)$ be the sequence constructed by \Cref{def:witness_sequence_construction}.
        Then, for all $i \geq 1$ with $w_i \neq \perp$, there is a success barrier for $w_i$.
    \end{lemma}
    \begin{proof}
        We prove the statement inductively.
        We start by proving that there is a success barrier for $w_1=(v_k,t_1)$.
        If $w_1$ is successful, there must be a directed path $P=((v_1,t'_1),\dots,(v_{k-1},t'_{k-1}))$ of length $k-1$ with $(v_{k-1},t'_{k-1})\in\pre_{G_W}(w_1)$ s.t. $P$ consists only of successful workers and connects $w_1$ to an initial worker $(v_1,t'_1)$.
        The topology of the worker graph yields, that $t'_1<\dots<t'_{k-1}<t_1$.
        Thus, there exists a successful worker $(v_j,t_j)$ assigned before time $t_1$ for each $1 \leq j < k$.
        The existence of a last successful worker $w''_j$ assigned to $v_j$ before time $t_1$ immediately follows.
        As $W$ only contains $(M_1,M'_1, w_1)$ at this point, it does not contain any failed workers. Therefore, the $w''_j$ are not linked at time $t$, and the success barrier for $w_1$ is given by $\{w''_1,\dots,w''_{k-1}\}$.
        
        If $w_1$ failed, we must be in a case where we just emptied $W$ (cf. \Cref{def:witness_sequence_construction} Case (1a)) as we pick the first starting point of the witness structure such that $w_1$ is successful.
        From Case (1a), we know that there is a successful worker $w$ assigned before time $t_1$ at task $v_{k-1}$.
        This successful worker again implies the existence of a path of length $k-1$ of successful workers connecting $w$ to the initial task.
        As $W$ only contains $(M_1,M_1', w_1)$ at this point, we can argue analogously to the case where $w_1$ is successful that the success barrier for $w_1$ exists.
        
        We now fix some $i\geq 1$ and prove that the existence of the success barrier for $w_i=(v_k,t_i)$ implies the existence of the success barrier for $w_{i+1}$.
        Due to the existence of the success barrier of $w_i$, we know that there is a successful worker assigned to all tasks $v_j$ before time $t_i \leq t_{i+1}$ with $1\leq j< k$ that is not linked at time $t_i$.
        If $w_i$ failed, we have $w_{i+1} = (v_{k-1}, t_{i+1})$, and the existence of a successful worker assigned to task $v_j$ before time $t_{i+1}$ directly follows for all $1\leq j< k-1$.
        In particular, there are also last successful workers $w''_j = (v_j, t''_j-1)$ assigned before time $t_{i+1}$ for all $1 \leq j < k-1$.
        Further, by \Cref{obs:non-decreasing}, there can be no failed workers in $W$ that are assigned between times $t_i$ and $t_{i+1}$.
        It follows that no $w''_j$ is linked at time $t_{i+1}$.  
        Similarly, if $w_i$ is successful, we have $w_{i+1} = (v_{k+1}, t_{i+1})$, $t_{i+1} > t_i$, and $w_i$ itself is a successful node located at task $v_k$, again yielding the existence of a successful worker assigned to task $v_j$ before time $t_{i+1}$ for all $1\leq j< k+1$.
        As before, this implies that there are also last successful workers $w''_j=(v_j,t''_j-1)$ assigned before time $t_{i+1}$ for all $1\leq j< k+1$.
        Further, by \Cref{obs:non-decreasing}, $w_i$ itself is not linked at time $t_{i+1}$.
        It again follows that no $w''_j$ is linked at time $t_{i+1}$.
        We conclude that there is a success barrier for $w_{i+1}$.
    \end{proof}
    
    We are now ready to prove that the sequence constructed by \Cref{def:witness_sequence_construction} is a valid witness sequence w.r.t. the assignment that caused the execution to fail.
    To this end, we show that the construction always terminates and that the constructed sequence satisfies \Cref{def:witness_sequence,def:valid_witness_sequence}.
    \begin{proof}[Proof of \Cref{lem:valid_witness_existence}]
        If there is a task $v\in V$ such that all workers assigned to task $v$ are malicious, the construction immediately terminates with $W = ((M_1, \emptyset ,\perp))$ and $M_1 = \{(v, 1), \dots, (v, \gamma)\}$, which is a witness sequence. Further, $\mathcal{M}=M_1$ contains only malicious workers, $\mathcal{M'}=\emptyset$ and $\mathcal{H}=\emptyset$.
        Thus, $W$ is also valid w.r.t. $A$.
        Otherwise, note that by \Cref{obs:non-decreasing}, the assignment times of the honest workers in the witness sequence are non-decreasing.
        Notably, in all cases of \Cref{def:witness_sequence_construction} except Case (2c) the assignment times are increasing.
        Let $(M_i,M'_i,w_i)$ with $w_i = (v_k, t)$ be a triple added to the sequence in Case (2c).
        By construction, we have $w_{i-1} = (v_{k+1}, t')$ for $t'\leq t$.
        Therefore Case $(2c)$ can occur no more than $n-1$ times in a row and the construction must eventually terminate.

        For a triple $(M_i,M'_i,w_i)$ added by the construction, $M_i$ and $w_i$ exactly match the structural requirements of a witness sequence in \Cref{def:witness_sequence}.
        Further, all workers in $M_i$ are clearly malicious and $w_i$ is an honest worker that failed if $w_i \in \mathcal{F}_W$ and was successful if $w_i \in \mathcal{H} \setminus \mathcal{F}$, i.e. $M_i$ and $w_i$ satisfy the requirements of a valid witness sequence w.r.t. $A$ in \Cref{def:valid_witness_sequence}. 
        It remains to prove that $M'_i$ also matches \Cref{def:witness_sequence,def:valid_witness_sequence}.
        If $w_i$ is successful, we set $M'_i = \emptyset$, which matches the requirements.
        On the other hand, if $w_i$ failed, we need to prove several properties.
        Let $w_i=(v_k,t)$.
        By \Cref{lem:success_barrier_existence}, there is a last successful worker $w''=(v_{k-1},t''-1)$ with $t''-1<t$.
        Thus, the set $M'_i = \{(v_{k-1},t''), \dots, (v_{k-1}, \min\{t''+2\delta-1,t_\mathrm{max}(v_{k-1})\})\}$ is well-defined.
        First, we show that indeed $t'' \leq t - 2\delta$ and $t\geq t_\mathrm{min}(v_k)+\delta$ hold, and that all workers in $M'_i$ are malicious.
        Then, we prove that $M'_i \cap M'_j = \emptyset$ for all $1 \leq j \leq \ell$ with $j \neq i$, and that $M'_i \cap M_j = \emptyset$ for all $1 \leq j \leq \ell + 1$.
        
        First, assume for contradiction that $t''> t-2\delta$.
        Then, $w_i\in\suc_{G_W}(w'')$ which contradicts that $w_i$ failed.
        Thus, $t''\leq t-2\delta$ has to hold.
        Further, since $t'' - 1 \geq t_{\mathrm{min}}(v_{k-1}) = t_{\mathrm{min}}(v_k) - \delta$, we get that $t \geq t'' + 2\delta > t_{\mathrm{min}}(v_k) + \delta$.
        Next, note that there can be no failed workers in $M'_i$, as $M'_i \subseteq \suc_{G_W}(w'')$.
        Additionally, since $t''\leq t-2\delta$, all workers in $M'_i$ were assigned before time $t$.
        Hence, there can be no successful worker in $M'_i$, as otherwise $w''$ would not be the last successful worker assigned to $v_{k-1}$ before time $t$.
        We get that $M'_i$ contains only malicious workers.

        Next, let $w_j = (v', t')$ with $j \neq i$ be a different failed worker in $W$ (if there is more than one).
        If $v' \neq v_k$, we immediately get $M'_i \cap M'_j = \emptyset$ because $M'_i$ and $M'_j$ belong to different tasks.
        Thus, consider the case that $v' = v_k$.
        If $t' \leq t''$, we get $M'_i \cap M'_j = \emptyset$ because all workers in $M'_j$ were assigned before time $t'$.
        Otherwise, if $t' > t''$, assume for contradiction that $M'_i \cap M'_j \neq \emptyset$.
        Since both $M'_i$ and $M'_j$ consist of only malicious workers and directly follow the last successful worker assigned to task $v_{k-1}$ before time $t$ and $t'$, respectively, we have $M'_i = M'_j$ and $w''$ is this last successful worker for both $w_i$ and $w_j$.
        But then, $w''$ is linked at either time $t$ or time $t'$, which contradicts that $w''$ is not linked by \Cref{lem:success_barrier_existence}.
        We conclude that $M'_i \cap M'_j = \emptyset$ has to hold.

        Finally, consider some $M_j$ with $1 \leq j \leq \ell + 1$.
        If $j = i$, we immediately get $M'_i\cap M_j=\emptyset$ as $M_i$ and $M'_i$ belong to different tasks.
        If $j > i$, note that no worker in $M_j$ was assigned before time $t$ by \Cref{obs:non-decreasing}.
        As all workers in $M'_i$ were assigned before time $t$, we get $M'_i\cap M_j=\emptyset$.
        Lastly, if $j < i$, assume for contradiction that $M'_i\cap M_j \neq \emptyset$.
        Since $M'_i$ contains only malicious workers and $M_j$ is followed by an honest worker $w_j = (v_{k-1}, t')$, we get that $t' > t''$ has to hold.
        To connect $w_j$ with $w_i$ in the sequence, there needs to be a successful worker assigned to task $v_{k-1}$ that was assigned between times $t''$ and $t$, because our construction only moves downwards when a successful worker is encountered and the assignment times of the honest workers are non-decreasing by \Cref{obs:non-decreasing}.
        But then, $w''$ would not be last successful worker assigned to task $v_{k-1}$ before time $t$, which is a contradiction.
        We conclude that $M'_i \cap M_j = \emptyset$ has to hold.
    \end{proof}

    \subsection{Witness Existence for the DAG Case}
    \label{sec:valid_witness_sequence_existence_dags}

    We now examine the more general setting in which the task graph $G = (V,E)$ is a DAG.
    As before, we begin by proving a few simple properties that will help us in arguing that the construction is well-defined.

    \begin{lemma}[Properties of Assignments in DAGs]\label{lem:assignment_properties_dags}
        Let $G=(V,E)$ be a task graph and let $G_W=(V_W,E_W)$ be its worker graph.
        For an assignment $A$ of workers to tasks, the following holds:
        \begin{itemize}
            \item If there is no task for which all assigned workers are malicious, the first honest worker assigned to any initial task is successful.
            \item If a worker $(v, t)$ failed, the set $\pre_G(v)$ is non-empty and $(v,t)$ failed to receive $v$'s input corresponding to at least one task in $\pre_G(v)$.
            \item If $A$ caused the computation to fail, there is a task where no assigned worker is successful.
        \end{itemize}
    \end{lemma}
    \begin{proof}
        If there is no task for which all assigned workers are malicious, each initial task is assigned at least one honest worker.
        All honest workers assigned to an initial task receive their input from the source and are thus successful.
        In particular, the first honest worker assigned to any initial task is successful.
        Next, consider a worker $(v, t)$ that failed.
        Since all honest workers assigned to any initial task are successful, we have $D(v)=k$ for $1 < k \leq n$.
        Hence, $\pre_G(v)$ is non-empty.
        For $(v, t)$ to have failed, it can not have received all of its inputs, i.e., there is a task in $\pre_G(v)$ that $(v,t)$ did not receive its input from.
        Finally, assume for contradiction that every task is assigned a worker that is successful.
        Then, the target receives the output of each final task.
        This contradicts $A$ causing the computation to fail.
    \end{proof}

    To extend our construction from \Cref{def:witness_sequence_construction} to the DAG case, we need to generalize the way we find a starting task for the construction, and the way we decide which predecessor and successor to move to.
    To this end, we need to keep track of how we traverse the DAG.
    To achieve this, we will employ a stack $S$.
    Whenever we leave a task $v$ moving upwards to a predecessor of $v$, we push task $v$ onto $S$. Whenever we decide which successor to move downwards to, we pop the top element of $S$ and move to that task.

    \begin{definition}[Valid Witness Sequence Construction for DAGs]\label{def:witness_sequence_construction_dags}
        For a task graph $G = (V,E)$ that is a DAG and its worker graph $G_W = (V_W, E_W)$, we consider a failed execution of the algorithm on $G$.
        Our construction again iteratively adds elements $(M,M',w)\in\mathcal{P}(V_W)\times\mathcal{P}(V_W)\times V_W$ to $W$, such that $M$ and $M'$ are sets of malicious workers, and $w$ is an honest worker, until a final element $(M,\emptyset,\perp)$ is added.
                If there is a task $v \in V$ such that all workers assigned to $v$ are malicious, we again add the triple $(M_1, \emptyset, \perp)$ to $W$ with $M_1 = \{(v, t_{\mathrm{min}}(v)), \dots, (v, t_{\mathrm{max}}(v))\}$ and end the construction.
        Otherwise, we employ a stack $S$ that we initialize to be empty and generalize the witness sequence for paths in the following way:

        \begin{itemize}
            \item To identify the starting task, we first consider the task $v_\mathrm{fail} \in V$ of smallest depth for which no assigned worker is successful (ties broken arbitrarily; such a worker exists, cf. \Cref{lem:assignment_properties_dags}).
            Each task of smaller depth than $D(v_\mathrm{fail})$ must have been assigned at least one successful worker.
            Starting with $v\gets v_\mathrm{fail}$, we iteratively move upwards in $G$ until we reach a task whose first worker is successful (which has to exist, cf. \Cref{lem:assignment_properties_dags}).
            If $v$'s first honest worker $(v,t)$ failed, we move towards the predecessor $v'$ of $v$, whose first successful worker was assigned last, i.e. $v' = \arg\max_{v_i \in \pre_G(v)}t_i$ (ties broken arbitrarily), where $(v_i, t_i)$ is the first successful worker assigned to $v$'s predecessor $v_i\in\pre_G(v)$. 
            Afterwards, we perform $S.\codefont{push($v$)}$ and set $v \leftarrow v'$.
            Once $v$'s first honest worker $(v,t)$ is successful, we add $(M_1, \emptyset , (v,t))$  to $W$ with $M_1 = \{(v, t_\mathrm{min}(v)), \dots, (v, t-1)\}$.

            \item After adding an element $(M_i,M'_i, w_i)$ to $W$, if $w_i = \perp$, our construction terminates.
            Otherwise, if $w_i = (v,t)$, we add $(M_{i+1},M'_{i+1}, w_{i+1})$ to $W$.
            To define $M_{i+1}$ and $w_{i+1}$ we distinguish two cases:
            \begin{enumerate}
                \item If $w_i$ is successful, we perform $v'\gets S.\codefont{pop}()$. $v'$ is a successor of $v$ due to \Cref{lem:dag_construction_well_defined}. We continue on $v'$ as in Case (1) of \Cref{def:witness_sequence_construction}.
                \item If $w_i$ failed, let $v' \in \pre_G(v)$ be an arbitrary preceding task from which $w$ received no correct output (such a task exists due to \Cref{lem:assignment_properties_dags}).
                We perform $S.\codefont{push(v)}$ and continue on $v'$ as in Case (2) of \Cref{def:witness_sequence_construction}.
            \end{enumerate}
            Additionally, if $w_{i+1}=\perp$ or $w_{i+1}$ is successful, we define $M'_{i+1}=\emptyset$, and if $w_{i+1}=(v,t')$ failed, we define $M'_{i+1} = \{(v'',t''), \dots, (v'', \min\{t''+2\delta-1,t_\mathrm{max}(v'')\})\}$, where $v''\in\pre_G(v)$ is an arbitrary preceding task from which $w_{i+1}$ received no output, and $(v'',t''-1)$ is the last successful worker assigned to $v''$ with $t''-1<t'$ (which exists due to \Cref{lem:success_barrier_existence_dags}).
        \end{itemize}
    \end{definition}

    Note that \Cref{obs:non-decreasing} extends to the DAG case. Thus, we will still reference it in the following analysis.
    Next, we argue that whenever we enter Case (1) of \Cref{def:witness_sequence_construction_dags}, the task on top of the stack $S$ will be a successor of the task we are currently at.
    Hence, we still move downwards in the DAG and it makes sense to continue as in Case (1) of the path case construction in \Cref{def:witness_sequence_construction}.

    \begin{lemma}[Construction Well-Defined]\label{lem:dag_construction_well_defined}
        When we encounter a worker $w_i=(v,t)$ that is successful (Case (1) of \Cref{def:witness_sequence_construction_dags}) there is a successor $v'$ of $v$ on top of $S$.
    \end{lemma}
    \begin{proof}
        We consider the top element of the stack when we enter a task $v\in V$ during the construction. 
        Note that entering $v$ from a predecessor requires $v$ to already have been on the stack (it got popped in Case (1) of \Cref{def:witness_sequence_construction_dags}).
        Thus, the first time we enter $v$, we enter from a successor $v'$ of $v$ (including while searching for the starting task), which pushes $v'$ on top of the stack.
        Hence, whenever we enter $v$ from a predecessor $v''$ of $v$, we must have visited $v$ before.
        Since we must have removed all elements we pushed to the stack since last visiting $v$, we conclude that some successor $v'$ of $v$ must be on top of the stack.
    \end{proof}

    We now generalize the concept of the success barrier for DAGs.
    We say an honest worker $w = (v,t)$ has a success barrier, if for every \emph{ancestor} $v''$ of $v$ in $G$ there is a successful worker assigned to $v$ before time $t$.
    Moreover, if $w'' = (v'', t''-1)$ is the last successful worker assigned to ancestor $v''$ before time $t$, there should be no failed worker $w' = (v', t')$ in $W$ with $v' \in \suc_G(v'')$ and $t'' - 1 < t' < t$, \emph{such that the sequence moved upwards from $v'$ to $v''$ after encountering $w'$}.
    As in the path case, this implies that $w''$ has not been ``used'' before by some other failed worker $w_{i'}$ in $W$ to define its own set $M'_{i'}$ for any $i' < i$.
    We formally define the success barrier for DAGs as follows.

    \begin{definition}[Success Barrier for DAGs]\label{def:success_barrier_dags}
        For a task graph $G=(V,E)$ and its worker graph $G_W=(V_W,E_W)$, we consider an assignment of workers to tasks which caused the computation to fail.
        Let $W$ be the witness sequence constructed by \Cref{def:witness_sequence_construction_dags}.
        We call a successful worker $w'' = (v'', t''-1)$ \emph{linked at time $t$} if there is a failed worker $w_i = (v', t')$ in $W$ such that $v' \in \suc_G(v'')$, $t''-1 < t' < t$ and $w_{i+1}$ is assigned to task $v''$.
        
        Let $(M,M',w)$ be a triple we add to $W$ (cf. \Cref{def:witness_sequence_construction_dags}) with $w = (v,t)$.
        If for all ancestors $v'$ of $v$ in $G$ there is a successful worker assigned to task $v'$ before time $t$ and the last such worker $w_{v'}$ is not linked at time $t$, we call $\{w_{v'}\mid v'\: \text{ancestor of } v\}$ the success barrier for $w$.
    \end{definition}

    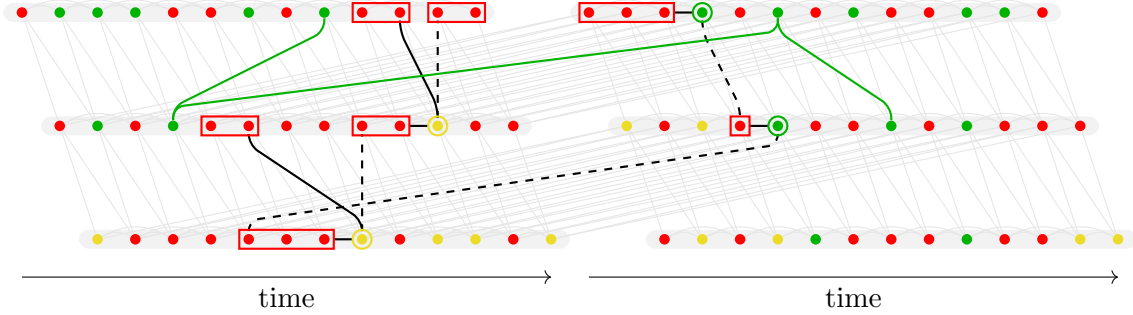
\begin{figure}[ht!]
        \begin{center}
            \begin{tikzpicture}[scale=0.5]
                \def\layers
                {
                    {{0,1,1,1,0,0,1,0,1,0,0,0,0},{0,0,0,1,0,1,0,1,0,0,1,1,0}},
                    {{0,1,0,1,0,0,0,0,0,0,2,0,0},{2,0,2,0,1,0,0,1,0,1,0,0,0}},
                        {{2,0,0,0,0,0,0,2,0,2,2,0,2},{0,2,0,2,1,0,0,0,1,0,0,2,2}}, 
                }
                \def\numoffset{1}
                \def\numsuccs{2}
                \def\numassign{13}

                \def\xoffset{2}
                \def\yscale{3}

                \foreach \layer [count=\y from 0] in \layers {
                    \pgfmathparse{dim({\layer})}
                    \let\layersize\pgfmathresult
                    \pgfmathparse{-(\layersize-1)/2}
                    \let\xstart\pgfmathresult
                    \foreach \task [count=\xx from 0] in \layer{
                        \pgfmathparse{\y+(\numassign+\xoffset)*(\xstart+\xx)}
                        \let\x\pgfmathresult
                        \fill[black!5,rounded corners] (\x,-\y*\yscale-0.25) rectangle (\x+\numassign,-\y*\yscale+0.25);
                        \foreach \worker [count=\w] in \task {
                            
                            \ifnum\worker=0
                                \node[circle,fill=red,inner sep=1.4pt] (\y-\xx-\w) at (\x+\w-0.5,-\y*\yscale) {};
                            \fi
                            \ifnum\worker=1
                                \node[circle,fill=green!70!black,inner sep=1.4pt] (\y-\xx-\w) at (\x+\w-0.5,-\y*\yscale) {};
                            \fi
                            \ifnum\worker=2
                                \node[circle,fill=yellow!90!black,inner sep=1.4pt] (\y-\xx-\w) at (\x+\w-0.5,-\y*\yscale) {};
                            \fi
                        }
                    }
                }

                \foreach \y/\x/\xx in {0/0/0,0/1/0,0/1/1,1/0/0,1/1/0,1/1/1} {
                    \pgfmathparse{int(\y+1)}
                    \let\yy\pgfmathresult
                    \foreach \w in {1,...,\numassign}{
                        \pgfmathparse{int(\w+1-\numoffset)}
                        \let\wmin\pgfmathresult
                        \pgfmathparse{int(\w+\numsuccs-\numoffset)}
                        \let\wmax\pgfmathresult
                        \foreach \ww in {\wmin,...,\wmax}{
                            \ifnum\ww>\numassign\else 
                                \draw[black!10] (\y-\x-\w) -- (\yy-\xx-\ww);
                            \fi
                        }
                    }
                }

                \foreach \layer [count=\y from 0] in \layers {
                    \foreach \task [count=\x from 0] in \layer{
                        \foreach \worker [count=\w] in \task {
                            \pgfmathparse{int(\w+1)}
                            \let\wmin\pgfmathresult
                            \pgfmathparse{int(\w+\numsuccs)}
                            \let\wmax\pgfmathresult
                            \foreach \ww in {\wmin,...,\wmax}{
                                \ifnum\ww>\numassign\else 
                                    \draw[black!10] (\y-\x-\w) to[bend left] (\y-\x-\ww);
                                \fi
                            }
                        }
                    }
                }

                \foreach \y/\x/\w/\xx in {0/1/4/1,1/1/5/0} {
                    \pgfmathparse{int(\y+1)}
                    \let\yy\pgfmathresult
                    \draw[thick,rounded corners,dashed] ($(\y-\x-\w)+(0,-0.25)$) -- ($(\y-\x-\w)+(0,-0.5)$) -- ($(\yy-\xx-\w)+(0,+0.5)$) -- ($(\yy-\xx-\w)+(0,+0.25)$);
                }
                
                \foreach \y/\x/\w/\xx in {2/0/8/0,1/0/11/0} {
                    \pgfmathparse{int(\y-1)}
                    \let\yy\pgfmathresult
                    \pgfmathparse{int(\w+1)}
                    \let\ww\pgfmathresult
                    \draw[thick,rounded corners,dashed] ($(\y-\x-\w)+(0,+0.25)$) -- ($(\y-\x-\w)+(0,+0.5)$) -- ($(\yy-\xx-\ww)+(0,-0.5)$) -- ($(\yy-\xx-\ww)+(0,-0.25)$);
                }
                
                \foreach \y/\x/\w/\xx/\ww in {2/0/8/0/6,1/0/11/0/11} {
                    \pgfmathparse{int(\y-1)}
                    \let\yy\pgfmathresult
                    \draw[thick,rounded corners] ($(\y-\x-\w)+(0,+0.25)$) -- ($(\y-\x-\w)+(0,+0.5)$) -- ($(\yy-\xx-\ww)+(0,-0.5)$) -- ($(\yy-\xx-\ww)+(0,-0.25)$);
                }
                \foreach \y/\x/\w/\ww in {1/0/5/6,0/0/10/11} {
                    \draw[red, thick] ($(\y-\x-\w)+(-0.25,-0.25)$) rectangle ($(\y-\x-\ww)+(+0.25,+0.25)$);
                }
                
                \foreach \y/\x/\w in {0/1/4,1/1/5} {
                    \draw[thick] ($(\y-\x-\w)+(-0.25,0)$) -- ($(\y-\x-\w)+(-0.75,0)$);
                }
                \foreach \y/\x/\w in {0/1/4,1/1/5} {
                    \draw[green!70!black, thick] (\y-\x-\w) circle (0.25);
                }
                
                \foreach \y/\x/\w in {2/0/8,1/0/11} {
                    \draw[thick] ($(\y-\x-\w)+(-0.25,0)$) -- ($(\y-\x-\w)+(-0.75,0)$);
                }
                \foreach \y/\x/\w in {2/0/8,1/0/11} {
                    \draw[yellow!90!black, thick] (\y-\x-\w) circle (0.25);
                }
                
                \foreach \y/\x/\w/\ww in {0/1/1/3,1/1/4/4,2/0/5/7,1/0/9/10,0/0/12/13} {
                    \draw[red, thick] ($(\y-\x-\w)+(-0.25,-0.25)$) rectangle ($(\y-\x-\ww)+(+0.25,+0.25)$);
                }

                \foreach \y/\x/\w/\xx/\ww in {1/0/4/0/9,1/0/4/1/6,1/1/8/1/6} {
                    \pgfmathparse{int(\y-1)}
                    \let\yy\pgfmathresult
                    \draw[green!70!black,thick,rounded corners] (\y-\x-\w) to ($(\y-\x-\w)+(0,+0.5)$) to ($(\yy-\xx-\ww)+(0,-0.5)$) to (\yy-\xx-\ww);
                }

                \draw[->] (-7,-2*\yscale-1) to node[midway,below] {time} (7,-2*\yscale-1); 
                \draw[->] (8,-2*\yscale-1) to node[midway,below] {time} (22,-2*\yscale-1); 

            \end{tikzpicture}
            \caption{An example of the success barrier for a DAG, where $\gamma = 13$ and $\delta = 1$. The green line connects the nodes in the success barrier for the first failed worker in the witness sequence.
            We can link each failed worker with a unique honest node in its success barrier and assigned to a preceding task.
            This honest worker must be succeeded by $2\delta=2$ malicious workers that are not in the witness sequence, which compose the set $M'$ corresponding to that failed worker.}
            \label{fig:success_barrier_dags}
        \end{center}
    \end{figure}

    \begin{lemma}[Success Barrier Existence for DAGs]\label{lem:success_barrier_existence_dags}
        For a task graph $G = (V,E)$ and its worker graph $G_W=(V_W,E_W)$, we consider an assignment of workers to tasks which caused the computation to fail.
        Let $W=((M_1,M'_1,w_1), (M_2,M'_2, w_2), \dots)$ be the sequence constructed by \Cref{def:witness_sequence_construction_dags}.
        Then, for all $i\geq 1$ with $w_i\neq\perp$, there is a success barrier for $w_i$.
    \end{lemma}
    \begin{proof}
        We prove the statement inductively.
        We start by proving that there is a success barrier for $w_1 = (v_1, t_1)$.
        If $w_1$ is successful, consider an arbitrary path $(v'_1, \dots, v'_\ell = v_1)$ in $G$ from an initial task $v'_1$ to $v_1$.
        Analogous to the path case, we can argue that there must be a directed path of successful workers $P = ((v'_1, t'_1), \dots, (v'_{\ell-1}, t'_{\ell-1}))$ in $G_W$ such that $(v'_{\ell-1}, t'_{\ell-1}) \in \pre_{G_W}(w_1)$ and $t'_1 < \dots < t'_{\ell - 1} < t_1$, as otherwise $w_1$ would not be successful (cf. \Cref{lem:success_barrier_existence}).
        Hence, since all ancestors of $v$ lie on some path from an initial task to $v$, there exists a successful worker assigned before time $t_1$ for each ancestor of $v$.
        As $W$ only contains $(M_1, w_1)$ at this point, it does not contain any failed workers. 
        Therefore, no successful worker is linked at time $t_1$ and the success barrier for $w_1$ exists.

        If $w_1$ failed, we must be in a case where we just emptied $W$ (cf. \Cref{def:witness_sequence_construction} Case (1a)), as we pick the first starting point of the witness structure such that $w_1$ is successful.
        From Case (1a), we know that there is a successful worker $w$ assigned before time $t_1$ at some preceding task $v \in \pre_G(v_1)$.
        This successful worker again implies that for every ancestor of $v$ there exists a successful worker assigned before time $t_1$.
        However, note that $v_1$ might have some ancestors that are no ancestors of $v$.
        We now argue that these ancestors also have such a successful worker.
        As the assignment times of the honest workers previously added to the sequence were non-decreasing by \Cref{obs:non-decreasing}, we know that task $v_1$ was not visited by the witness sequence before.
        Since we only go downwards to tasks that are in the stack $S$, and only add tasks to $S$ after the witness sequence has visited them (once the starting task has been fixed), this implies that $v_1$ belongs to the tasks that were added to $S$ while determining the (original) starting task of the construction.
        As $v_1$ was on top of $S$ when the witness sequence was at task $v$, we must have gone up from task $v_1$ to $v$ while determining the starting task.
        Therefore, $v$ is the predecessor of $v_1$ whose first successful worker was assigned last.
        As there is a successful worker assigned to $v$ before time $t_1$, this implies that for all other predecessors of $v_1$ there is also successful worker assigned before time $t_1$.
        As before, this also implies that for every ancestor of $v_1$ there is a successful worker assigned before time $t_1$.
        As $W$ only contains $(M_1, M'_1, w_1)$ at this point, we can argue analogously to the case where $w_1$ is successful that the success barrier for $w_1$ exists.

        We now fix some $i \geq 1$ and prove that the existence of the success barrier for $w_i = (v_i, t_i)$ implies the existence of the success barrier for $w_{i+1} = (v_{i+1}, t_{i+1})$.
        Due to the existence of the success barrier of $w_i$, we know that for each ancestor $v'$ of $v_i$ there is a last successful worker assigned to $v'$ before time $t_{i}$ that is not linked at time $t_i$.
        From \Cref{obs:non-decreasing}, we have $t_{i+1} \geq t_i$, which implies that for each ancestor of $v_i$ there is also a successful worker assigned before time $t_{i+1}$.
        Further, there can be no failed worker in $W$ strictly between times $t_i$ and $t_{i+1}$.
        Hence, it still holds that for each ancestor $v'$ of $v_i$ the last successful worker assigned to $v'$ before time $t_{i+1}$ is not linked at time $t_{i+1}$.
        
        If $w_i$ failed, we have $v_{i+1} \in \pre_G(v_i)$ and, thus, every ancestor of $v_i$ is also an ancestor of $v_{i+1}$.
        Therefore, it immediately follows that there is a success barrier for $w_{i+1}$. 
        If $w_i$ is successful, we have $v_{i+1} \in \suc_G(v_i)$, $t_{i+1} > t_i$, and $w_i$ itself is a successful node that is not linked at time $t_{i+1}$.
        Hence, for $v_i$ and all ancestors $v'$ of $v_i$, there is a last successful worker assigned before time $t_{i+1}$ that is not linked at time $t_{i+1}$.
        However, $v_{i+1}$ might have some ancestors that are no ancestors of $v_i$.
        To analyze these, we distinguish two cases.
        
        First, consider the case that $W$ has not visited task $v_{i+1}$ before, i.e. there is no honest worker $(v_j, t_j)$ in $W$ with $j < i$ and $v_j = v_{i+1}$.
        Then, we can argue analogously to the base case where $w_1$ failed, that both $v_i$ and $v_{i+1}$ belong to the tasks that were added to $S$ while determining the starting task of the construction, and thus, for every ancestor of $v_{i+1}$ there is a successful worker assigned before time $t_{i+1}$.
        Further, note that whenever $W$ moves upwards from a task $v$, it only traverses ancestors of $v$ until it reaches $v$ again.
        Therefore, as the ancestors of $v_{i+1}$ that are no ancestors of $v_i$ were not added to $S$ while determining the starting task of the construction, they have never been visited by $W$ before.
        In particular, at the time that $w_{i+1}$ is added to $W$, there is no failed worker in $W$ for all ancestors of $v_{i+1}$ that are no ancestors of $v_i$, i.e. no successful worker of these tasks is linked at time $t_{i+1}$.
        We conclude that there is a success barrier for $w_{i+1}$.
        
        Next, consider the case that $W$ has visited task $v_{i+1}$ before.
        To this end, let $w_j = (v_j, t_j)$ be the last honest worker in $W$ with $j < i$ and $v_j = v_{i+1}$.
        By the induction hypothesis we know that every ancestor $v'$ of $v_j$ has a last successful worker assigned to $v'$ before time $t_{j}$ that is not linked at time $t_j$.
        Since $v_j = v_{i+1}$ and $t_j < t_{i+1}$, it follows that also for every ancestor of $v_{i+1}$ there is a successful worker assigned before time $t_{i+1}$.
        Further, note that since $v_{i+1}$ was on top of the stack $S$ after $w_i$ was added to $W$, we have $v_{j+1} \in \pre_G(v_j)$.
        Therefore, $v_{j+1} = v_i$ has to hold, as otherwise $w_j$ would not be the last honest worker in $W$ assigned to task $v_{i+1}$ before $w_{i+1}$.
        Therefore, at the time that $w_{i+1}$ is added to $W$, the last successful worker assigned before time $t_{i+1}$ is still not linked for all ancestors of $v_{i+1}$ that are no ancestors of $v_i$.
        We conclude that there is a success barrier for $w_{i+1}$.
    \end{proof}

    Finally, we prove that the construction in \Cref{def:witness_sequence_construction_dags} indeed produces a valid witness sequence.
    Given the existence of the success barrier, the proof is very similar to the one of \Cref{lem:valid_witness_existence} in the path case.
    For the sake of completeness, we repeat the arguments here anyways.

    \begin{proof}[Proof of \Cref{lem:valid_witness_existence_dags}]
    
        If there is a task $v\in V$ such that all workers assigned to task $v$ are malicious, the construction immediately terminates with $W = ((M_1, \emptyset ,\perp))$ and $M_1 = \{(v, 1), \dots, (v, \gamma)\}$, which is a witness sequence. Further, $\mathcal{M}=M_1$ contains only malicious workers, $\mathcal{M'}=\emptyset$ and $\mathcal{H}=\emptyset$.
        Thus, $W$ is also valid w.r.t. $A$.
        Otherwise, note that by \Cref{obs:non-decreasing}, the assignment times of the honest workers in the witness sequence are non-decreasing.
        Notably, in all cases of \Cref{def:witness_sequence_construction} (which is referred to by \Cref{def:witness_sequence_construction_dags}) except Case (2c) the assignment times are increasing.
        Let $(M_i,M'_i,w_i)$ with $w_i = (v, t)$ be a triple added to the sequence in Case (2c).
        By construction, we have $w_{i-1} = (v', t')$ for $v'\in\pre_G(v)$ $t'\leq t$.
        Therefore Case (2c) can occur no more than $D-1$ times in a row and the construction must eventually terminate.

        For a triple $(M_i,M'_i,w_i)$ added by the construction, $M_i$ and $w_i$ exactly match the structural requirements of a witness sequence in \Cref{def:witness_sequence}.
        Further, all workers in $M_i$ are clearly malicious and $w_i$ is an honest worker that failed if $w_i \in \mathcal{F}_W$ and was successful if $w_i \in \mathcal{H} \setminus \mathcal{F}$, i.e. $M_i$ and $w_i$ satisfy the requirements of a valid witness sequence w.r.t. $A$ in \Cref{def:valid_witness_sequence}. 
        It remains to prove that $M'_i$ also matches \Cref{def:witness_sequence} and \Cref{def:valid_witness_sequence}.
        If $w_i$ is successful, we set $M'_i = \emptyset$, which matches the requirements.
        On the other hand, if $w_i$ failed, we need to prove several properties.
        Let $w_i=(v,t)$.
        By \Cref{lem:success_barrier_existence}, there is a last successful worker $w_{v''}=(v'',t''-1)$ s.t.  $v'\in\pre_G(v)$ is the arbitrary preceding task we picked for $M'_i$ (cf. \Cref{def:witness_sequence_construction_dags}) and $t''-1<t$.
        Thus, the set $M'_i = \{(v'',t''), \dots, (v'', \min\{t''+2\delta-1,t_\mathrm{max}(v'')\})\}$ is well-defined.
        First, we show that indeed $t'' \leq t - 2\delta$ and $t\geq t_\mathrm{min}(v)+\delta$ hold, and that all workers in $M'_i$ are malicious.
        Then, we prove that $M'_i \cap M'_j = \emptyset$ for all $1 \leq j \leq \ell$ with $j \neq i$, and that $M'_i \cap M_j = \emptyset$ for all $1 \leq j \leq \ell + 1$.
        
        First, assume for contradiction that $t''> t-2\delta$.
        Then, $w_i\in\suc_{G_W}(w'')$ which contradicts that $w_i$ failed.
        Thus, $t''\leq t-2\delta$ has to hold.
        Further, since $t'' - 1 \geq t_{\mathrm{min}}(v_{k-1}) = t_{\mathrm{min}}(v_k) - \delta$, we get that $t \geq t'' + 2\delta > t_{\mathrm{min}}(v_k) + \delta$.
        Next, note that there can be no failed workers in $M'_i$, as $M'_i \subseteq \suc_{G_W}(w'')$.
        Additionally, since $t''\leq t-2\delta$, all workers in $M'_i$ were assigned before time $t$.
        Hence, there can be no successful worker in $M'_i$, as otherwise $w''$ would not be the last successful worker assigned to $v''$ before time $t$.
        We get that $M'_i$ contains only malicious workers.

        Next, let $w_j = (v', t')$ with $j \neq i$ be a different failed worker in $W$ (if there is more than one).
        If $v' \neq v$, we immediately get $M'_i \cap M'_j = \emptyset$ because $M'_i$ and $M'_j$ belong to different tasks.
        Thus, consider the case that $v' = v$.
        If $t' \leq t''$, we get $M'_i \cap M'_j = \emptyset$ because all workers in $M'_j$ were assigned before time $t'$.
        Otherwise, if $t' > t''$, assume for contradiction that $M'_i \cap M'_j \neq \emptyset$.
        Since both $M'_i$ and $M'_j$ consist of only malicious workers and directly follow the last successful worker assigned to task $v''$ before time $t$ and $t'$, respectively, we have $M'_i = M'_j$ and $w''$ is this last successful worker for both $w_i$ and $w_j$.
        But then, $w''$ is linked at either time $t$ or time $t'$, which contradicts that $w''$ is not linked by \Cref{lem:success_barrier_existence}.
        We conclude that $M'_i \cap M'_j = \emptyset$ has to hold.

        Finally, consider some $M_j$ with $1 \leq j \leq \ell + 1$.
        If $j = i$, we immediately get $M'_i\cap M_j=\emptyset$ as $M_i$ and $M'_i$ belong to different tasks.
        If $j > i$, note that no worker in $M_j$ was assigned before time $t$ by \Cref{obs:non-decreasing}.
        As all workers in $M'_i$ were assigned before time $t$, we get $M'_i\cap M_j=\emptyset$.
        Lastly, if $j < i$, assume for contradiction that $M'_i\cap M_j \neq \emptyset$.
        Since $M'_i$ contains only malicious workers and $M_j$ is followed by an honest worker $w_j = (v', t')$, we get that $t' > t''$ has to hold.
        To connect $w_j$ with $w_i$ in the sequence, there needs to be a successful worker assigned to task $v'$ that was assigned between times $t''$ and $t$, because our construction only moves downwards when a successful worker is encountered and the assignment times of the honest workers are non-decreasing by \Cref{obs:non-decreasing}.
        But then, $w''$ would not be last successful worker assigned to task $v'$ before time $t$, which is a contradiction.
        We conclude that $M'_i \cap M_j = \emptyset$ has to hold.
    \end{proof}

    \subsection{Correctness and Work Bound}\label{sec:correctness_and_work_bound}
    
    It remains to show that it is unlikely for a random assignment $A$ of workers to tasks to contain a witness sequence that is valid w.r.t. $A$.
    We start by counting the number of malicious assignments $A$ has to perform for a witness sequence $W$ to be valid w.r.t. $A$.
    Intuitively, each downwards movement from a task $v$ to a task $v' \in \suc_G(v)$ increases the number of malicious workers in $W$ by at least $\delta-1$, as the assignments to $v'$ are delayed by $\delta$ rounds relative to $v$.
    By the same argument, the number of malicious workers decreases by $\delta$ for each upwards movement.
    However, as each upwards movement is associated with a set $M'$ consisting of $2\delta$ additional malicious workers, we still get an overall increase of $\delta$ malicious workers.
    We provide a formal bound with the following two lemmas.
    
    \begin{lemma}[Additional Malicious Workers]
    \label{lem:additional_malicious_workers}
        For a task graph $G=(V,E)$ and its corresponding worker graph $G_W=(V_W,E_W)$, let $A$ be an assignment of workers to tasks and let $W=((M_1,M'_1,w_1),\dots,(M_\ell,M'_\ell,w_\ell),\allowbreak (M_{\ell+1},M'_{\ell+1},w_\perp))$ be a witness sequence that is valid w.r.t. $A$.
        If $w_i = (v,t)$ failed for $1 \leq i \leq \ell$, we have $|M'_{i}| \geq \min\{2\delta, t_\mathrm{max}(v)-t+\delta+1\}$.
    \end{lemma}
    \begin{proof}
        Since $w_i$ failed, we have $t \geq t_{\mathrm{min}}(v) + \delta$ and $M'_i = \{(v',t''),\dots,(v',\min\{t''+2\delta-1,t_\mathrm{max}(v')\})\}$ for a $v'\in\pre_G(v)$ and $t''\leq t- 2\delta$ (cf. \Cref{def:witness_sequence}, Property (2)).
        If $t'' + 2\delta -1 \leq t_\mathrm{max}(v')$, we have $|M_i'| = (t'' + 2 \delta-1) - t'' + 1 = 2\delta$.
        Otherwise, if $t'' + 2\delta -1 > t_\mathrm{max}(v')$, we have $|M_i'| = t_\mathrm{max}(v') - t'' + 1$.
        As we also have $t_\mathrm{max}(v') = t_\mathrm{max}(v) - \delta$ and $t''\leq t- 2\delta$, we obtain
        \begin{equation*}
            |M'_i| = t_\mathrm{max}(v') - t'' + 1 \geq t_\mathrm{max}(v) - \delta - (t-2\delta) + 1 = t_\mathrm{max}(v) - t + \delta + 1.
        \end{equation*}
    \end{proof}

    \begin{lemma}[Total Number of Malicious Workers]\label{lem:num_malicious_workers}
        For a task graph $G=(V,E)$ and its worker graph $G_W=(V_W,E_W)$, we consider an assignment $A$ of workers to tasks.
        Further, Let $W = ((M_1,M_1', w_1),\allowbreak \dots,\allowbreak (M_\ell,M'_\ell, w_\ell),\allowbreak (M_{\ell+1},\emptyset, \perp))$ be a witness sequence that is valid w.r.t. $A$.
        Then, there are at least $m\geq\max\{(\delta-1)\cdot \ell,\gamma\}$ malicious workers in $\mathcal{M}_W \cup \mathcal{M}'_W$.
    \end{lemma}
    \begin{proof}
        We prove the statement by induction over $\ell$.
        For $\ell = 0$, we have $W = ((M_1,\emptyset,\perp))$ and $M_1 = \{(v,t_\mathrm{min}(v)),\allowbreak \dots,\allowbreak (v,t_\mathrm{max}(v))\}$ for some task $v \in V$.
        Hence, there are at least $\gamma$ malicious workers in $M_1 = \mathcal{M}_W$ and the lemma holds.
        Now, for an arbitrary, but fixed $\ell$, we have that for each valid witness sequence $W$ containing $\ell$ honest workers there are at least $\max\{(\delta-1)\cdot \ell,\gamma\}$ malicious workers in $\mathcal{M}_W \cup \mathcal{M}'_W$.
        We extend our claim to the case where $W = ((M_1,M'_1, w_1), \dots, (M_{\ell+1},M'_{\ell+1},w_{\ell+1}), (M_{\ell+2},\emptyset,\perp))$ contains exactly $\ell+1$ honest workers.
        To this end, let $w_{\ell+1} \coloneqq (v,t)$ and consider the valid witness sequence $\overline{W} = ((M_1,M_1',w_1),\allowbreak \dots,\allowbreak (M_{\ell},M_{\ell}', w_{\ell}),\allowbreak (M_{\ell+1} \cup \overline{M},\emptyset,\perp))$, where we define $\overline{M} \coloneqq \{(v,t),\allowbreak \dots,\allowbreak (v,t_\mathrm{max}(v))\}$ extends $M_{\ell+1}$ with malicious workers as far as possible to the right.
        By the induction hypothesis, we have that $\mathcal{M}_{\overline{W}} \cup \mathcal{M}'_{\overline{W}}$ contains at least $\max\{(\delta-1)\cdot \ell,\gamma\}$ malicious workers.
        We now consider two cases.

        If $w_{\ell+1} = (v,t)$ was successful, we have $M_{\ell+2} = \{(v',t+1),\allowbreak \dots,\allowbreak (v',t_\mathrm{max}(v'))\}$ for some $v' \in \suc_G(v)$.
        From $t_\mathrm{max}(v') = t_\mathrm{max}(v) + \delta$, it follows that 
        \[
            |M_{\ell+2}| = t_\mathrm{max}(v') - (t+1) + 1 = t_\mathrm{max}(v) - t + \delta = |\overline{M}| + \delta - 1\text{.}
        \]
        Due to $\mathcal{M}_{W} = (\mathcal{M}_{\overline{W}} \cup M_{\ell + 2}) \setminus \overline{M}$ and $\mathcal{M}'_{W} = \mathcal{M}'_{\overline{W}}$ and as all involved sets are disjoint, we have $|\mathcal{M}_W \cup \mathcal{M}'_W| = |\mathcal{M}_{\overline{W}} \cup \mathcal{M}'_{\overline{W}}| + |M_{\ell+2}|-|\overline{M}|$.
        Hence, $|\mathcal{M}_W \cup \mathcal{M}'_W| = |\mathcal{M}_{\overline{W}} \cup \mathcal{M}'_{\overline{W}}|+\delta-1$.
        As $\mathcal{M}_{\overline{W}} \cup \mathcal{M}'_{\overline{W}}$ contains at least $\max\{(\delta-1)\cdot \ell,\gamma\}$ malicious workers, we conclude that $\mathcal{M}_W \cup \mathcal{M}'_W$ contains at least $\max\{(\delta-1)\cdot (\ell+1),\gamma\}$ malicious workers.

        If $w_{\ell+1} = (v,t)$ failed, we know that $|M'_{\ell+1}| \geq \min\{2\delta, t_\mathrm{max}(v)-t+\delta+1\}$ by \Cref{lem:additional_malicious_workers}.
        We distinguish between two subcases.
        If $t > t_\mathrm{max}(v) - \delta$, then $M_{\ell+2} = \emptyset$ and $ \min\{2\delta, t_\mathrm{max}(v)-t+\delta+1\}=t_\mathrm{max}(v)-t+\delta+1$.
        Thus, we have $|M'_{\ell+1}| \geq t_\mathrm{max}(v) - t + \delta + 1 = |\overline{M}| + \delta$.
        Otherwise, if $t \leq t_\mathrm{max}(v) - \delta$, we have $M_{\ell+2} = \{(v', t), \dots, (v', t_\mathrm{max}(v'))\}$ for some $v' \in \pre_G(v)$
        and $\min\{2\delta, t_\mathrm{max}(v)-t+\delta+1\}=2\delta$.
        Thus, we have $|M'_{\ell+1}| \geq 2\delta$.
        From $t_\mathrm{max}(v') = t_\mathrm{max}(v) - \delta$, it follows that 
        \[
            |M_{\ell+2}|+|M'_{\ell+1}| \geq t_\mathrm{max}(v') - t + 1 + 2\delta = t_\mathrm{max}(v) - t + 1 + \delta = |\overline{M}| + \delta.
        \]
        Due to $\mathcal{M}_{W} = (\mathcal{M}_{\overline{W}} \cup M_{\ell + 2}) \setminus \overline{M}$ and $\mathcal{M}'_{W} = \mathcal{M}'_{\overline{W}} \cup M'_{\ell+1}$ and as all involved sets are disjoint, we have $|\mathcal{M}_W \cup \mathcal{M}'_W| = |\mathcal{M}_{\overline{W}} \cup \mathcal{M}'_{\overline{W}}| + |M_{\ell+2}|+|M'_{\ell+1}|-|\overline{M}|$.
        Hence, in both subcases, $|\mathcal{M}_W \cup \mathcal{M}'_W|\geq|\mathcal{M}_{\overline{W}} \cup \mathcal{M}'_{\overline{W}}|+\delta$.
        As $\mathcal{M}_{\overline{W}} \cup \mathcal{M}'_{\overline{W}}$ contains at least $\max\{(\delta-1)\cdot \ell,\gamma\}$ malicious workers, we conclude that $\mathcal{M}_W \cup \mathcal{M}'_W$ contains at least $\max\{(\delta-1)\cdot (\ell+1),\gamma\}$ malicious workers.
    \end{proof}

    Lemma \ref{lem:num_malicious_workers} gives us a lower bound on the number of malicious assignments that need to happen for a fixed witness sequence to be valid.
    This immediately yields the following upper bound for the probability of a fixed witness sequence being valid w.r.t. a random assignment of workers to tasks.
    
    \begin{corollary}[Probability of a Valid Witness Sequence]\label{cor:witness_probability}
         For a task graph $G=(V,E)$ and its worker graph $G_W = (V_W, E_W)$, let $W = ((M_1,M_1', w_1), \dots, (M_\ell,M'_\ell, w_\ell), (M_{\ell+1},\emptyset, \perp))$ be a fixed witness sequence in $G_W$. 
        Further, let $A$ be a random assignment of workers to tasks, where each worker is independently malicious with probability at most $\beta$ and honest with probability at least $(1- \beta)$.
        The probability that $W$ is valid w.r.t $A$ is at most $\beta^m$ where $m=|\mathcal{M}|+|\mathcal{M}'| \geq\max\{(\delta-1)\cdot |\mathcal{H}|,\gamma\}$.
    \end{corollary}

    Next, we count the total number of witness sequences with fixed sizes $|\mathcal{M}|$, $|\mathcal{H}|$ and $|\mathcal{F}|$.

    \begin{lemma}[Number of Witness Sequences]\label{lem:witness_amount_dags}
        For a task graph $G= (V,E)$ and its worker graph $G_W = (V_W, E_W)$, the number of witness sequences with fixed sizes $|\mathcal{M}|$, $|\mathcal{H}|$ and $|\mathcal{F}|$ in $G_W$ is at most $ n\cdot \binom{|\mathcal{M}|+|\mathcal{H}|}{|\mathcal{H}|}\cdot (2\gamma d)^{|\mathcal{H}|}$, where $d\geq\max\{\indeg(G),\outdeg(G)\}$.
    \end{lemma}
    \begin{proof}
        There are $n$ possible tasks at which a witness sequence can start.
        Further, for each of the $|\mathcal{H}|$ many honest workers in a witness sequence, there are at most $d$ many options to which task the sequence moves next.
        Therefore, there are at most $n \cdot d^{|\mathcal{H}|}$ many sequences of tasks a witness sequence can traverse.
        Given such a sequence of tasks, we count the number of witness sequences that follow this sequence of tasks.
        To this end, we encode each triple $(M,M',w)$ of a witness sequence in the following way.
        First, we represent $(M,w)$ as a string on $\{0,1,2\}^*$ that starts with a number of $0$-entries equal to $|M|$ followed by a $1$ if $w$ is a successful worker, a $2$, if $w$ is a failed worker, or no additional symbol, if $w=\perp$.
        Next, we encode $M'$ simply by the time $t$ when the first worker in $M'$ was assigned, which implicitly also fixes the size of $M'$.
        Concatenating the resulting strings of all triples yields a unique representation of the witness sequence.

        We now count the number of $\{0,1,2\}$-strings with exactly $|\mathcal{M}|$ $0$-entries and exactly $|\mathcal{H}|$ $1$-or-$2$-entries, of which exactly $|\mathcal{F}|$ are a $2$-entries.
        There are $\binom{|\mathcal{M}|+|\mathcal{H}|}{|\mathcal{H}|}$ ways to place the $1$-or-$2$-entries in the string and $\binom{|\mathcal{H}|}{|\mathcal{F}|}$ ways to choose which of them are a $2$-entry.
        Further, there are at most $\gamma$ options for when the first worker of a set $M'$ can be assigned, yielding another factor of $\gamma^{|\mathcal{F}|}$.
        In total, we obtain that the number of witness sequences with fixed sizes $|\mathcal{M}|$, $|\mathcal{H}|$, and $|\mathcal{F}|$ is at most
        \[
            n\cdot d^{|\mathcal{H}|}\cdot \binom{|\mathcal{M}|+|\mathcal{H}|}{|\mathcal{H}|} \cdot \binom{|\mathcal{H}|}{|\mathcal{F}|} \cdot \gamma^{|\mathcal{F}|} \leq n\cdot \binom{|\mathcal{M}|+|\mathcal{H}|}{|\mathcal{H}|} \cdot (2\gamma d)^{|\mathcal{H}|}\text{.}
        \]
    \end{proof}

    Finally, we show that it is unlikely for a valid witness sequence to exist by combining \Cref{cor:witness_probability} and \Cref{lem:witness_amount_dags}.
    This proves that, w.h.p., the computation will succeed.

    \begin{lemma}[Correctness]\label{lem:witness_correctness_dags}
        Let $G$ be a task graph with in- and outdegree bounded by $d$.
        If $\delta\geq\max\{2,\allowbreak\nicefrac{4}{(\alpha^2\ln^2\nicefrac{1}{\beta})},\allowbreak\nicefrac{(2\log_{\nicefrac{1}{\beta}}(2e\gamma d))}{\alpha}+1\}$ and $\gamma\geq\nicefrac{(c+5)}{(1-\alpha)}\log_{\nicefrac{1}{\beta}} n$ for any constant parameter $0<\alpha<1$, the algorithm terminates correctly on $G$, w.h.p.
    \end{lemma}
    \begin{proof}        
        Let $A$ be a random assignment of workers to tasks, where each worker is independently malicious with probability at most $\beta$ and honest with probability at least $(1- \beta)$.
        We show that no witness sequence in $G_W$ is valid w.r.t $A$, w.h.p.
        From the contraposition of \Cref{lem:valid_witness_existence_dags}, it follows that $A$ does not cause the computation to fail, w.h.p., which proves the lemma.
        
        Before we start, we establish a helpful bound for an arbitrary constant parameter $0<\alpha<1$ if we pick $\delta$ as above:
        \[
            \alpha \geq \frac{\log_{\nicefrac{1}{\beta}}(2e\delta\gamma d)}{\delta-1} \tag{$\star$}
        \]
        To reduce the dependency on $\gamma$ and $d$, we do this in two steps, i.e., we separately prove $\nicefrac{\alpha}{2} \geq \frac{\log_{\nicefrac{1}{\beta}}(2e\gamma d)}{\delta-1}$ and $\nicefrac{\alpha}{2} \geq \frac{\log_{\nicefrac{1}{\beta}}(\delta)}{\delta-1}$ for different bounds on $\delta$.
        We can directly solve the first inequality for $\delta$:
        \[
            \nicefrac{\alpha}{2} \geq \frac{\log_{\nicefrac{1}{\beta}}(2e\gamma d)}{\delta-1} \Longleftrightarrow \delta \geq \frac{2\log_{\nicefrac{1}{\beta}}(2e\gamma d)}{\alpha}+1
        \]
        For the second inequality, we use $\delta\geq\max\{2,\nicefrac{4}{(\alpha^2\ln^2\nicefrac{1}{\beta})}\}$ and $\sqrt{x}>\ln x$ for $x > 0$ to obtain:
        \[
            \delta \geq \frac{4}{\alpha^2\ln^2\nicefrac{1}{\beta}} \Longrightarrow \sqrt{\delta} \geq \frac{2}{\alpha\ln{\nicefrac{1}{\beta}}} \Longrightarrow \alpha \ln \nicefrac{1}{\beta}\geq  \frac{2}{\sqrt{\delta}} = \frac{2\sqrt{\delta}}{\delta} = \frac{\sqrt{\delta}}{\delta-\nicefrac{\delta}{2}} \overset{\delta \geq 2}{\geq} \frac{\sqrt{\delta}}{\delta-1}\geq \frac{\ln(\delta)}{\delta -1}
        \] 
        Combining the inequalities yields ($\star$) for $\delta\geq\max\{2,\nicefrac{4}{(\alpha^2\ln^2\nicefrac{1}{\beta})},\allowbreak\nicefrac{(2\log_{\nicefrac{1}{\beta}}(2e\gamma d))}{\alpha}+1\}$.
    
        Additionally, \Cref{lem:num_malicious_workers} yields $|\mathcal{H}|\leq\nicefrac{m}{(\delta-1)}$ and $m\geq\gamma$. We employ the union bound to combine the results of \Cref{cor:witness_probability} and \Cref{lem:witness_amount_dags} obtaining for the probability that any witness sequence $W_{|\mathcal{M}|,|\mathcal{H}|,|\mathcal{F}|}$ with fixed sizes$|\mathcal{M}|$, $|\mathcal{H}|$ and $|\mathcal{F}|$ that is valid w.r.t. $A$ exists.
        \begin{align*}
            \Pr[\exists W_{|\mathcal{M}|,|\mathcal{H}|,|\mathcal{F}|} \text{ valid w.r.t. }A] &\leq n\cdot \binom{|\mathcal{M}|+|\mathcal{H}|}{|\mathcal{H}|} \cdot (2\gamma d)^{|\mathcal{H}|}\cdot\beta^m            \\
            &\leq n\cdot\binom{m+\nicefrac{m}{(\delta-1)}}{\nicefrac{m}{(\delta-1)}}\cdot (2\gamma d)^{\nicefrac{m}{(\delta-1)}}\cdot \beta^m\tag{$|\mathcal{M}|\leq m$, $|\mathcal{H}|\leq\nicefrac{m}{(\delta-1)}$}\\
            &\leq n\cdot \left(\frac{e(m+\nicefrac{m}{(\delta-1)})}{\nicefrac{m}{(\delta-1)}}\right)^{\nicefrac{m}{(\delta-1)}}\cdot (2\gamma d)^{\nicefrac{m}{(\delta-1)}}\cdot \beta^m\\
            &= n\cdot \left(e\delta\right)^{\nicefrac{m}{(\delta-1)}}\cdot (2\gamma d)^{\nicefrac{m}{(\delta-1)}}\cdot \beta^m\\
            &= n\cdot \left(2e\delta\gamma d\right)^{\nicefrac{m}{(\delta-1)}}\cdot \beta^m\\
            &= n\cdot (\nicefrac{1}{\beta})^{(\nicefrac{m}{(\delta-1)})\cdot\log_{\nicefrac{1}{\beta}}\left(2e\delta\gamma d\right)}\cdot (\nicefrac{1}{\beta})^{-m}\\
            &= n\cdot (\nicefrac{1}{\beta})^{-m+m\cdot(\nicefrac{\log_{\nicefrac{1}{\beta}}\left(2e\delta\gamma d\right)}{(\delta-1)})}\\
            &\leq n \cdot (\nicefrac{1}{\beta})^{-m(1-\alpha)}\tag{$\star$}\\
            &\leq n \cdot (\nicefrac{1}{\beta})^{-\gamma(1-\alpha)}\tag{$m\geq\gamma$}\\
            &\leq n \cdot (\nicefrac{1}{\beta})^{-(c+5) (\log_{\nicefrac{1}{\beta}} 2) \log n} \tag{$\gamma\geq\nicefrac{(c+5)\log_{\nicefrac{1}{\beta}}2}{(1-\alpha)}\log n$}\\
            &= n \cdot n^{-c-5}\\ 
            &\leq n^{-c-4}
        \end{align*}
        We apply the union bound three more times to extend the high probability result to all possible values of $|\mathcal{M}|$, $|\mathcal{H}|$, and $|\mathcal{F}|$. As none of these sets can contain more than $\gamma n$ workers, we obtain the probability for any valid witness sequence $W$ existing:
        \[
            \Pr[\exists W \text{ valid w.r.t. A}] \leq \sum_{|\mathcal{M}|} \sum_{|\mathcal{H}|} \sum_{|\mathcal{F}|} \Pr[ \exists W_{|\mathcal{M}|,|\mathcal{H}|,|\mathcal{F}|} \text{ valid w.r.t }A] \leq (\gamma n)^3 \cdot n^{-c-4} \leq n^{-c}
        \]
        using $\gamma^3\leq n$ for a sufficiently large $n$ as $\gamma=O(\log n)$.
    \end{proof}

    Lastly, we prove that, on expectation, only very few honest workers have to execute the task they are assigned to.
    Therefore, the expected total work performed by honest workers is very small.

    \begin{lemma}[Work Bound]\label{lem:manyworkers_dag_work}
        Let $G=(V,E)$ be a task graph. For any execution of the algorithm and any task $v\in V$, if $\delta\geq\nicefrac{((c+2)\log_{\nicefrac{1}{\beta}}\log n)}{2}$ and $\gamma=O(\log_{\nicefrac{1}{\beta}} n)$ then the total work performed by honest workers assigned to $v$ is $1+o(1)$ in expectation. Furthermore, the total work performed by all honest workers is $n(1+o(1))$ in expectation.
    \end{lemma}
    \begin{proof}
    Consider any execution of the algorithm and any fixed task $v$. Let the random variable $X_v$ denote the number of times an honest worker assigned to $v$ executes task $v$. If all honest workers assigned to $v$ fail, $X_v=0$. Otherwise, let $w_0$ be the first successful honest worker assigned to $v$. If after $w_0$ there is no contiguous sequence of $2\delta$ adversarial workers then all honest workers chosen after $w_0$ can adopt the output of $w_0$, which implies that $X_v =1$. Otherwise, $X_v$ is upper bounded by $\gamma$. Since the probability that there is a contiguous sequence of $2\delta$ adversarial workers is at most $\gamma \cdot \beta^{2\delta} \le \gamma/(\log n)^{c+2}$, it follows that
    \[
    \E[X_v] \le 1 + \gamma \cdot \gamma/(\log n)^{c+2} = 1+o(1)
    \]
    Let the random variable $X$ be the total number of task executions of honest workers. Then
    \[
    \E[X] = \sum_v \E[X_v] = n(1+o(1))
    \]
    which completes the proof of the lemma.
    \end{proof}

    We conclude this section by proving our main theorem.
    
    \begin{proof}
        We start by fixing the parameter $\delta\geq\max\{2,\allowbreak\nicefrac{8e}{\alpha^2\ln^2\nicefrac{1}{\beta}},\allowbreak\nicefrac{(2\log_{\nicefrac{1}{\beta}}(2e\gamma d))}{\alpha}+1,\nicefrac{((c+1)\log_{\nicefrac{1}{\beta}}\log n)}{2}\}$ as well as the parameter $\gamma\geq\nicefrac{(c'+5)}{(1-\alpha)}\log_{\nicefrac{1}{\beta}} n$ for any constant parameter $0<\alpha<1$ (ensuring that $\delta=O(\log_{\nicefrac{1}{\beta}} (d) \cdot \log_{\nicefrac{1}{\beta}}\log n)$ and $\gamma=O(\log_{\nicefrac{1}{\beta}} n)$).
        The correctness follows from \Cref{lem:witness_correctness_dags}.
        The runtime follows from the fact that the supervisor assigns the final worker to any task of depth $D$ in round $(D-1)\cdot \delta+\gamma=O(D\log_{\nicefrac{1}{\beta}} (d) \cdot \log_{\nicefrac{1}{\beta}}\log n+\log_{\nicefrac{1}{\beta}} n)$.
        By construction, each honest worker performs at most one computation for each time it was assigned.
        The verification and communication bounds follow from the fact that each worker interacts with $O(d\delta)=O(d\log_{\nicefrac{1}{\beta}} (d) \cdot \log_{\nicefrac{1}{\beta}}\log n)$ other workers.
        The bound for task executions and total work was proven in \Cref{lem:manyworkers_dag_work}.
        For each initial task, the source needs to send the input to the $\gamma$ workers assigned to it and for each final task, the target needs to receive the output from each of the $\gamma$ workers assigned to it and verify each of them at most once.
        Finally, the supervisor assigns $\gamma=O(\log n)$ workers to each of the $n$ tasks, introduces each worker up to $d\delta=O(d\log_{\nicefrac{1}{\beta}} (d) \cdot \log_{\nicefrac{1}{\beta}}\log n)$ times, and introduces the source $O(\log n)$ times for each initial task. 
        This yields $O(n\log n)$ assignments and $O(nd\log n \cdot \log_{\nicefrac{1}{\beta}} (d) \cdot \log_{\nicefrac{1}{\beta}}\log n)$ introductions.
    \end{proof}


\section{Conclusion}
In this paper, we demonstrate for the first time how to handle a setting with a majority of malicious workers.
Instead of performing a majority vote, we assign logarithmically many workers to each task one by one.
By having the workers provide their output to their successors in the same task and in the next tasks, we reduce the work overhead for the honest workers -- down to an almost optimal overhead of $1+o(1)$ for each worker. 
The workers employ a lightweight verification approach to validate the output of their predecessors. 
We pipeline the assignment of workers to tasks of different levels to ensure that the computation terminates quickly.

Our work raises several important questions that warrant further study. 
Our focus is on problems where the claimed output for an individual task can be verified easily.
Sorting and matrix multiplication can be decomposed into such verifiable tasks~\cite{augustine2025supervised}.
However, we believe that many more problems could potentially be decomposed in a similar fashion and benefit from the supervised distributed computing framework.
Additionally, it would be interesting to explore whether cryptographic verification mechanisms, such as SNARKs and SNARGs, or hardware based verification mechanisms~\cite{TEE} can be leveraged within our framework. Since our supervisor is very lightweight, we believe that an automated but reliable process (such as a smart contract) could emulate its role, though this requires careful investigation.
Finally, studying heterogeneous DAGs, where some tasks are easily verifiable while others require more care, presents a promising direction.

\renewcommand{\bibname}{References} 
\bibliographystyle{plain}

\newpage
\bibliography{references}

\appendix

\newpage

\section{Pseudocode}\label{sec:pseudocode}

\begin{algorithm}[H]
    \DontPrintSemicolon
    \caption{Assigning many workers (Supervisor)}
    \label{alg:manyworkers_supervisor}
    
    \ForEach{$v\in V$}{
        $t_\mathrm{min}(v) \gets (D(v)-1) \cdot \delta + 1$ \tcp*{first round of assignments to $v$}
        
        $t_\mathrm{max}(v) \gets (D(v)-1) \cdot \delta + \gamma$ \tcp*{last round of assignments to $v$}
    }
    
    \For{round $t\gets1$ \KwTo $\gamma+ (D - 1) \cdot \delta$}{
        $\psi(t) \coloneqq \{v\in V \mid t \in [t_\mathrm{min}(v), t_\mathrm{max}(v)]\}$
        
        \ForEach{$v \in \psi(t)$}{
            Assign worker $w$ to $(v,t)$ independently, with probability $\beta$ to be malicious and probability $1-\beta$ to be honest
                
            Instruct workers assigned to $(v,t-2\delta),\dots,(v,t-1)$ to send output of $v$ to $w$
                
            \If{$w$ fails to verify all these outputs}{
                \If{$v$ is an initial task}{
                    Instruct the source to send the input of $v$ to $w$
                }
                \Else{
                    Instruct workers assigned to $(v',t-2\delta),\dots,(v',t-1)$ with $v' \in \pre_G(v)$ to send input of $v$ to $w$
                }
            }
            
        }
    }
\end{algorithm}

\begin{algorithm}[H]
    \DontPrintSemicolon
    \caption{Assigning many workers (Worker assigned to task $v$)}
    \label{alg:manyworkers_worker}
    
    \upon{receiving output candidates of $v$}{
        \ForEach{output candidate $o$}{
            \If{Verification of $o$ succeeds}{
                success $\gets$ true
                
                break
            }
        }
        \If{success = true}{
            Remain in system for $2\delta$ rounds to provide output to later workers
        }
        \Else{
            Notify supervisor of failure to verify outputs
        }
    }
    \upon{receiving input candidates of $v$ from workers assigned to $v'\in\pre_G(v)$}{
        \ForEach{input candidate $i$}{
            \If{Verification of $i$ succeeds}{
                success[$v'$] $\gets$ true
                
                break
            }
        }
        \If{success[$v'$] = true for each $v'\in\pre_G(v)$}{
            Compute output

            Remain in system for $2\delta$ rounds to provide output to later workers
        }
        \Else{
            Remain in system for $2\delta$ rounds to provide error to later workers
        }
    }
\end{algorithm}

\begin{algorithm}[H]
\DontPrintSemicolon
\caption{Supervised Path Computation (Supervisor's Logic)}
\label{alg:supervised_path}

\KwIn{Number of tasks $n$, \textit{Source}, \textit{Target} and a black-box to assign the workers}
\BlankLine
\tcc{Supervisor State Initialization}
\BlankLine
$P \gets \text{array of size } n \text{, initialized to NULL}$ \tcp*{Stores assigned workers}
$i \gets 1$ \tcp*{Current task index}

\BlankLine
\Proc{SupervisedPathComputation()}{
    \texttt{AssignNewWorker(1)} \tcp*{Start with the first task}

    \While{$i \le n$}{
        $worker\_id \gets P[i]$\;
        $message \gets \text{WaitForMessage}(worker\_id)$\;

        \uIf{$message$ is TIMEOUT}{
            \texttt{AssignNewWorker(i)} \tcp*{Worker failed, re-assign same task}
        }
        \uElseIf{$message$ is DONE}{
            \uIf{$i < n$}{
                $i \gets i + 1$ \tcp*{Proceed to next task}
                \texttt{AssignNewWorker(i)}\;
            }
            \Else{
                $\text{IntroduceWorkerToTarget}(P[n], \text{Target})$\;
                $target\_message \gets \text{WaitForMessage}(\text{Target})$\;
                \uIf{$target\_message$ is DONE}{
                    $i \gets i + 1$ \tcp*{Success! Terminate loop}
                }
                \Else{
                    \texttt{AssignNewWorker(n)} \tcp*{Target rejected $n$-th worker}
                }
            }
        }
        \ElseIf{$message$ is REJECT}{
            \If{$i > 1$}{
                $i \gets i - 1$ \tcp*{Rollback to previous task}
            }
            \texttt{AssignNewWorker(i)} \tcp*{Re-assign worker for rolled-back task}
        }
    }
   
}
 \BlankLine
\SetKwFunction{AssignNewWorker}{AssignNewWorker}

\SetKwProg{Fn}{Function}{:}{}
\Fn{\AssignNewWorker{$i$}}{
    $new\_worker \gets \text{SampleWorker()}$\;
    $P[i] \gets new\_worker$\;
    \uIf{$i = 1$}{
        $\text{IntroduceWorkerToSource}(new\_worker, \text{Source})$\;
    }
    \Else{
        $previous\_worker \gets P[i-1]$\;
        $\text{IntroduceWorkers}(previous\_worker, new\_worker)$\;
    }
}
\BlankLine
\end{algorithm}

\newpage

\section{Improved Bound and Infeasibility Results}\label{sec:bound_and_infeasibility}

In this section, we provide some additional analysis for the algorithms presented in \cite{augustine2025supervised}.
In \Cref{ssec:improved_path_bound}, we improve the allowed fraction of malicious workers for the runtime proof of the path algorithm from $\beta\leq \nicefrac{1}{12}$ to $\beta<\nicefrac{1}{2}$.
In \Cref{ssec:infeasibility_path}, we prove that this new bound for $\beta$ is tight.
More specifically, for any $\beta>\nicefrac{1}{2}$, the path algorithm requires exponential time to terminate, w.h.p.
Finally, we prove a similar result for the DAG algorithm in \Cref{ssec:infeasibility_dag}.
The authors of \cite{augustine2025supervised} provide a bound $\beta\leq(\nicefrac{1}{(2(2d+1))})^{2+\varepsilon}$ for the DAG algorithm to terminate within $O(\nicefrac{(D+\log n)}{\varepsilon})$, w.h.p.
We almost match this bound by proving that the DAG algorithm requires exponential time to terminate for $\beta>\nicefrac{\ln c}{c^2}$ for $c=\nicefrac{(d+1)}{2}$.

\subsection{Improved Bound for Path Graphs}
\label{ssec:improved_path_bound}
We start by presenting a new analysis of the path algorithm from \cite{augustine2025supervised}.
Specifically, we improve the bound on the number of malicious workers under which the algorithm still terminates correctly in linear time from $\beta \leq \nicefrac{1}{12}$ to $\beta < \nicefrac{1}{2}$.

In the path algorithm, the supervisor assigns workers to tasks sequentially and processes their responses.
A worker assigned to $v_i$ must reply with either \texttt{DONE} (execution completed) or \texttt{REJECT} (input invalid or missing).
If a worker replies \texttt{DONE}, the supervisor proceeds to $v_{i+1}$ (or to the target if $i=n$).
If the reply is \texttt{REJECT}, the supervisor removes the workers assigned to $v_{i-1}$ and $v_i$ from the system and reassigns a new worker to $v_{i-1}$.
In case the supervisor does not receive \texttt{DONE}/\texttt{REJECT} within some specified time, it reassigns a new worker to $v_i$.
The computation terminates once the target replies with \texttt{DONE}.
The pseudocode of the path algorithm is provided as \Cref{alg:supervised_path}.

\begin{theorem}\label{thm:path_case_majority_honest}
    For a path graph of size $n$ it holds that: If $\beta < \nicefrac{1}{2}$, the supervised computation  terminates correctly under any adversarial strategy in $\nicefrac{n}{(1-2\beta)}+\Theta\left(\nicefrac{\sqrt{n \log n}}{(1-2\beta)}\right)$ rounds, w.h.p. Furthermore, the source just needs to send the input $\nicefrac{1}{(1-2\beta)}$ times, and the target just needs to receive the solution $\nicefrac{1}{(1-2\beta)}$ times, in expectation.
\end{theorem}
\begin{proof}
According to \cite{augustine2025supervised}, always sending \texttt{REJECT} is a dominant adversarial strategy.
Thus, we prove the statement for that strategy and immediately obtain the general result.
When the always-\texttt{REJECT} strategy is used, each adversary is immediately rolled back after assignment.
Thus, whenever the supervisor makes an assignment to a task $v_i$ for $1<i\leq n$, the worker assigned to $v_{i-1}$ is guaranteed to be honest.

Let $X_t$ be a binary random variable representing the progress we make in round $t$.
If the supervisor assigns an honest worker, the computation proceeds to the next task and we have $X_t=+1$.
The supervisor assigning a malicious worker results in a rollback and we have $X_t=-1$.
Further, let $S_T=\sum_{t=1}^T X_t$ be the net progress we make in $T$ rounds.
The supervised computation terminates if it reaches the target, i.e., if the net progress $X_T$ reaches $n+1$.
Our goal is to prove that the computation terminates after $T=\nicefrac{n}{1-2\beta}+\Delta$ rounds, w.h.p., where we will choose $\Delta>0$ later to control the failure probability.
We start by computing the expected progress we make in each round:
\[
    \E[X_t] = (+1)\cdot \Pr[X_t=+1] + (-1) \cdot \Pr[X_t=-1] = (1-\beta)-\beta = 1-2\beta
\]
By linearity of expectation, we have $\E[S_T]=\sum_{t=1}^T \E[X_t] = T(1-2\beta)$.

We will employ a variant on Hoeffding's inequality~\cite{hoeffding1963probability} that states for the sum $S_T$ of independent random variables $X_1,\dots,X_T$ with $a_i\leq X_i\leq b_i$ for $1\leq i\leq T$:
\[
    \Pr[S_T-\E[S_T]\leq x]\leq \exp\left(-\frac{2x^2}{\sum_{i=1}^T(b_i-a_i)^2}\right)
\]
In our case, we have $a_i=-1$ and $b_i=+1$ for $1\leq i\leq T$, which allows us to simplify the bound:
\[
    \Pr[S_T-\E[S_T]\leq x]\leq \exp\left(-\frac{2x^2}{4T}\right) = \exp\left(-\frac{x^2}{2T}\right)
\]
Our goal is to bound the failure probability. We transform the failure event to match the bound:
\begin{align*}    
    S_T< n+1 &\Longleftrightarrow S_T\leq n\\
    &\Longleftrightarrow S_T-\E[S_T]\leq n-\E[S_T]\\
    &\Longleftrightarrow S_T-\E[S_T]\leq n-\left(\frac{n}{1-2\beta}+\Delta\right)(1-2\beta)\\
    &\Longleftrightarrow S_T-\E[S_T]\leq -\Delta(1-2\beta)
\end{align*}
In preparation of our application of Hoeffding's Inequality, we solve the following inequality for $\Delta$ using the quadratic formula.
\begin{align*}
    &\;\; \frac{(\Delta(1-2\beta))^2}{2(\frac{n}{1-2\beta}+\Delta)} \geq c\ln n\\
    \Longleftrightarrow&\;\; (\Delta(1-2\beta))^2 \geq 2c\ln n \cdot \left(\frac{n}{1-2\beta}+\Delta\right)\\
    \Longleftrightarrow&\;\; (1-2\beta)^2\cdot\Delta^2 -2c\ln n \cdot \Delta -\frac{2cn\ln n}{1-2\beta}\geq 0\\
    \Longleftrightarrow&\;\; \Delta\geq \frac{2c\ln n\pm\sqrt{(2c\ln n)^2+4(1-2\beta)^2\frac{2cn\ln n}{1-2\beta}}}{2(1-2\beta)^2}\\
    \Longleftrightarrow&\;\; \Delta\geq \frac{c\ln n\pm\sqrt{(c\ln n)^2+(1-2\beta)^2\frac{2cn\ln n}{1-2\beta}}}{(1-2\beta)^2}
\end{align*}
Note that we are only interested in the positive root.
As we want to pick $\Delta$ as small as possible, we have $\Delta =\Theta\left(\nicefrac{\sqrt{n\log n}}{(1-2\beta)}\right)$.
We are now ready to apply Hoeffdings inequality choosing $\Delta$ as above:
\[
    \Pr[S_T< n+1] \leq \exp\left(-\frac{(\Delta(1-2\beta))^2}{2(\frac{n}{1-2\beta}+\Delta)}\right) \leq \exp(-c\ln n) \leq n^{-c}
\]

It remains to prove the expected number of send and receives for the source and the target.
Recall that $E[X_t] = 1-2\beta$ is the expected progress a worker make in round $t$ for a task.
Therefore, it requires $\nicefrac{1}{(1-2\beta)}$ rounds to progress a task in expectation. We conclude that the source needs to send the input $\nicefrac{1}{(1-2\beta)}$ times and the target needs to receive the output $\nicefrac{1}{(1-2\beta)}$ times, in expectation.
\end{proof}

\subsection{Infeasibility Result for Path Graphs}\label{ssec:infeasibility_path}
In this section, we prove that the bound from \Cref{thm:path_case_majority_honest} is tight.
More specifically, we prove that for any $\beta>\nicefrac{1}{2}$, the path algorithm requires exponential time to terminate, w.h.p.
To achieve this, we fix an adversarial strategy and describe the resulting computation with a Markov chain that has two absorbing states representing the source and the target.
Using the well-established gambler's ruin problem (see, e.g. \cite{ruined_gambler}), we prove that once the computation reaches the first task, it is exponentially more likely to return to the source than to reach the target. 
Finally, we conclude that an exponential number of such trials is required until the target is reached for the first time, w.h.p.

\begin{theorem}\label{thm:lower_bound_linear_time}
    For a path graph of size $n$ it holds that: If $\beta > \nicefrac{1}{2}$, there is an adversarial strategy s.t. the supervised computation requires exponential time to terminate correctly, w.h.p.
\end{theorem}

\begin{proof}
    We consider the adversarial strategy of always sending a \texttt{REJECT}.
    We recall from the proof of \Cref{thm:path_case_majority_honest} that we can model the progress we make in each round $t$ under this adversarial strategy with a random variable $X_t$ s.t. $X_t=+1$, if the supervisor assigns an honest worker in round $t$, and $X_t=-1$, if the supervisor assigns a malicious worker in round $t$.
    Notably, these variables are independent and identically distributed.
    
    We describe the computation with a Markov chain with states $0,\dots,n+1$ that describe the net progress.
    State $0$ represents the computation being located at the source, state $i$ represents the computation being located at task $i$ for $1\leq i\leq n$, and state $n+1$ represents the computation having reached the target.
    As discussed above, the chain has transition probabilities $p\coloneqq\Pr[i\to i+1]=1-\beta$ and $q\coloneqq\Pr[i\to i-1]=\beta$ for $1\leq i\leq n$.
    Our first goal is to bound the probability $f_{1,n+1}$ that the chain reaches state $n+1$ before reaching state $0$ after reaching $1$ for the first time.
    Thus we choose the initial state to be $1$ (which is reasonable, as the supervisor will eventually assign an honest worker to the first task, w.h.p.), and we define the states $0$ and $n+1$ to be absorbing, i.e., $\Pr[0\to0]=1$ and $\Pr[1\to1]=1$.
    
    This is an instance of the classic gambler's ruin problem for which there are numerous sources, e.g., \cite{ruined_gambler}.
    We obtain the following identity for $f_{1,n+1}$ from the literature:
    \[
        f_{1,n+1} = \frac{1-\nicefrac{q}{p}}{1-\left(\nicefrac{q}{p}\right)^{n+1}}
    \]
    We have $\nicefrac{q}{p}=\nicefrac{\beta}{(1-\beta)}$, which increases monotonously on [0,1).
    $\beta>\nicefrac{1}{2}$ yields $\nicefrac{q}{p}>1$.
    For such fractions, we have for any $n\geq 1$:    \begin{align*}
        &\;\; (\nicefrac{q}{p})^n>1\\
        \Longleftrightarrow&\;\; -(\nicefrac{q}{p})^n<-1\\
        \Longleftrightarrow&\;\; (\nicefrac{q}{p})^{n+1}-(\nicefrac{q}{p})^n<(\nicefrac{q}{p})^{n+1}-1\\
        \Longleftrightarrow&\;\; (\nicefrac{q}{p})^n(\nicefrac{q}{p}-1)<\left(\nicefrac{q}{p}\right)^{n+1}-1\\
        \Longleftrightarrow&\;\; \frac{\nicefrac{q}{p}-1}{\left(\nicefrac{q}{p}\right)^{n+1}-1}<\left(\nicefrac{q}{p}\right)^{-n}
    \end{align*}
    We bound $f_{1,n+1}$ as follows:
    \[
        f_{1,n+1} = \frac{1-\nicefrac{q}{p}}{1-\left(\nicefrac{q}{p}\right)^{n+1}} =  \frac{\nicefrac{q}{p}-1}{\left(\nicefrac{q}{p}\right)^{n+1}-1}< \left(\nicefrac{q}{p}\right)^{-n}
    \]
    We consider an application of the gambler's ruin problem to be a single trial to reach the target before reaching the source again.
    It remains to prove that an exponential number of trials is required to reach the target, w.h.p.
    Let $X$ be a random variable representing the number of independent trials until state $n+1$ is reached for the first time.
    Choosing the exponential threshold $T=n^{-c}\cdot(\nicefrac{q}{p})^n$, we obtain:
    \begin{align*}
        \Pr[X\leq T]&=1-(1-f_{1,n+1})^T\\
        &\leq 1-(1-(\nicefrac{q}{p})^{-n})^T\\
        &\leq 1-(1-T\cdot(\nicefrac{q}{p})^{-n}) \tag{$(1-x)^T\geq 1-Tx$ for $x\in[0,1]$}\\
        &=T\cdot(\nicefrac{q}{p})^{-n}\\
        &=n^{-c}
    \end{align*}
    
\end{proof}

\subsection{Infeasibility Result for DAGs}\label{ssec:infeasibility_dag}
Finally, we prove that the bound of $\beta\leq(\nicefrac{1}{(2(2d+1))})^{2+\varepsilon}$ for the DAG algorithm from \cite{augustine2025supervised} is almost tight, too. 
More specifically, we prove that the DAG algorithm requires exponential time to terminate, w.h.p., if $\beta>\nicefrac{\ln c}{c^2}$, where $c=\nicefrac{(d+1)}{2}$ and $c\geq5$.
In the DAG algorithm, the supervisor keeps track of a set $F\subseteq V$ of tasks that it considers to be finished. 
In each round, it assigns workers to all tasks that are not in $F$, but where all predecessors are in $F$.
A worker assigned to a task $v$ must reply either with \texttt{DONE} (execution completed) or \texttt{REJECT(}$R$\texttt{)}, where $R\subseteq V$ is a subset of predecessors of $v$ from which the worker received an invalid or no input.
If a worker replies \texttt{DONE}, the supervisor adds its task to $F$.
If a worker replies \texttt{REJECT(}$R$\texttt{)}, the supervisor removes its task, every task in $R$ and every task reachable from any task in $R$ from $F$. 
Additionally, it removes all workers assigned to any of these tasks from the system.
If the supervisor receives neither \texttt{DONE} nor \texttt{REJECT}, it removes the worker that failed to reply according to the protocol from the system.
If the predecessors of the corresponding task are still in $F$ by the next round, the supervisor will assign a new worker to it.
The computation terminates once the target replies with \texttt{DONE}.

For our proof, we describe a family of worst-case DAGs, where the DAG algorithm struggles to make progress.
Employing several worst-case assumptions, we set up an instance of the gambler's ruin problem again, and are thus able to argue similarly to \Cref{thm:lower_bound_linear_time}.

\begin{theorem}
    If $\beta>\nicefrac{\ln c}{c^2}$ where $c\geq5$ and $c=\nicefrac{(d+1)}{2}$, then there is a DAG $G$ with $\indeg(G)\leq d$ and $\outdeg(G)\leq d$ where the protocol requires exponential time to terminate for some adversarial strategy, w.h.p.  
\end{theorem}
\begin{proof}
    We consider the DAG $G=(V,E)$ with $n$ levels of $c^2$ tasks, i.e., $V= \{0,\dots,n-1\}\times\{0,\dots,c^2-1\}$.
    We divide each level into $c$ blocks of $c$ nodes, i.e., block $j$ of level $i$ contains nodes $(i,c\cdot j),\dots,(i,c\cdot(j+1)-1)$ for $0\leq j\leq c-1$.
    For every level $i$ with $0\leq i< n-1$, block $j$ of level $i$ is connected to block $j$ of level $i+1$ for $0\leq j\leq c-1$ with a complete bipartite graph.                 Additionally, for $0\leq i< n-1$, $0\leq j\leq c-1$, and $0\leq k\leq c-1$, the $k$th task in block $j$ of level $i$ $(i,c\cdot j+k)$ is connected to the $k$th task of every block $j'$ of level $i+1$ $(i+1,c\cdot j'+k)$ with $0\leq j'\leq c-1$ and $j'\neq j$.                     \Cref{fig:dag_infeasibility} depicts the connections of two adjacent levels of the DAG for $c=3$.

    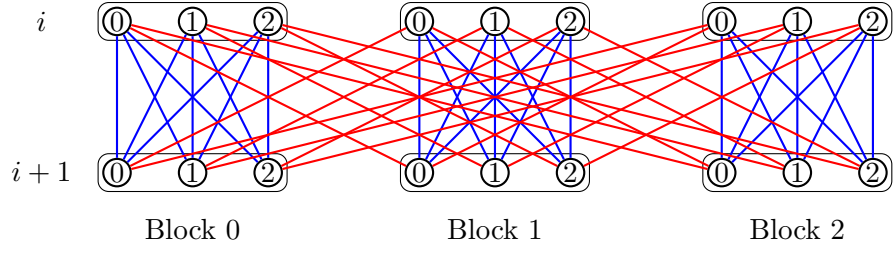
\begin{figure}[ht!]
        \begin{center}
            \begin{tikzpicture}
                \node at (0,0) {$i$};
                \node at (0,-2) {$i+1$};

                \foreach \offset [count = \oo from 0] in {0,4,8} {
                    \foreach \y [count = \yy] in {0,-2} {
                        \foreach \x [count = \xx from 0] in {1,...,3} {
                            \node[draw,circle,thick] (\oo-\xx-\yy) at (\x+\offset,\y) {};
                            \node at (\x+\offset,\y) {\xx};
                        }
                        \draw[black,rounded corners] (0.75+\offset,-0.25+\y) rectangle (3.25+\offset,0.25+\y);
                    }
                    \node at (\offset+2, -2.75) {Block \oo};
                }

                \foreach \o in {0,1,2} {
                    \foreach \x in {0,1,2} {
                        \foreach \t in {0,1,2} {
                            \draw[blue,thick] (\o-\x-1) -- (\o-\t-2);
                        }
                    }
                }
                \foreach \o in {0,1,2} {
                    \foreach \x in {0,1,2} {
                        \foreach \t in {0,1,2} {
                            \pgfmathparse{\o != \t}
                            \ifnum\pgfmathresult=1
                                \draw[red,thick] (\o-\x-1) -- (\t-\x-2);
                            \fi
                        }
                    }
                }
            \end{tikzpicture}
            \caption{The connections of two adjacent levels of the DAG for $c=3$. The blue edges connect tasks in the same block of levels $i$ and $i+1$ with a complete bipartite graph. The red edges connect each task in level $i$ with the corresponding task in each other block in level $i+1$.}
            \label{fig:dag_infeasibility}
        \end{center}
    \end{figure}

    It remains to prove that the protocol requires exponential time to terminate on $G$ for some adversarial strategy, w.h.p.
    Our goal is to set up another instance of the gambler's ruin problem. 
    To this end, we have all adversaries follow the strategy to always REJECT all their inputs upon assignment.
    For this strategy, the assignment of a single adversarial worker to block $j$ in level $i+1$ causes a REJECT to exactly one task in every other block $j'$ for $j'\neq j$ of level $i$ and all tasks in block $j$ of level $i$.
    This results in the supervisor having to reassign the rejected workers, and all their successors.
    Specifically, at least all workers in block $j$ of level $i$ and all workers of level $i+1$ get reassigned.

    We model the problem with a Markov chain $M$ with states $0,\dots,n+1$.
    $M$ is in state $i$, if $i$ is the last level where any honest workers are assigned.
    As a worst case assumption, we consider all workers assigned to level $i$ to be honest.
    Each transition of $M$ corresponds to at least one round of the supervised computation.

    For $M$ to increase its state ($i\to i+1$), all tasks in level $i+1$ have to receive an honest assignment as any malicious assignment in level $i+1$ would result in a reject and thus the removal of all workers in level $i+1$.
    Thus, we obtain using $\beta>\nicefrac{\ln c}{c^2}$:
    \[
        p\coloneqq\Pr[i\to i+1]=(1-\beta)^{c^2}\leq e^{-\beta c^2}<\nicefrac{1}{c}\text{.}
    \]
    For $M$ to decrease its state ($i\to i-1$), we require a process of at least two rounds.
    Firstly, at least one task in level $i+1$ has to receive a malicious assignment, resulting in a rollback of at least one block in level $i$.
    Secondly, at least one task in the block of $i$ that was just rolled back has to receive a malicious assignment, resulting in a rollback of the rest of level $i$.
    We obtain using $\beta>\nicefrac{\ln c}{c^2}$:
    \begin{align*}        
        q\coloneqq\Pr[i\to i-1]&\geq(1-(1-\beta)^{c^2})(1-(1-\beta)^c)\\
        &\geq(1-e^{-\beta c^2})(1-e^{-\beta c})\\
        &> (1-\nicefrac{1}{c})(\beta c-\nicefrac{(\beta c)^2}{2})\tag{$1-e^{-x}\geq x-\nicefrac{x^2}{2}$ for $x\geq 0$}\\
        &>(1-\nicefrac{1}{c})\cdot(\nicefrac{\ln c}{c}-\nicefrac{\ln^2 c}{2c^2})\tag{$\beta c< 1$}\\
    \end{align*}
    Note that $M$ has some self loops, where the state neither decreases nor increases. 
    These can be ignored for the gambler's ruin problem as they do not change the number of steps required to reach any absorbing state.
    Thus, we are only interested in the ratio $\nicefrac{q}{p}$.
    We have
    \[
        \nicefrac{q}{p}> \frac{(1-\nicefrac{1}{c})\cdot(\nicefrac{\ln c}{c}-\nicefrac{\ln^2 c}{2c^2})}{\nicefrac{1}{c}}=(c-1) \frac{\ln c-\nicefrac{\ln^2 c}{2c}}{c}\overset{c\geq 5}{>}1\text{.}
    \]
    We conclude that there is a drift towards the source.
    We can prove that the protocol requires exponential time to terminate, w.h.p., for the always-\texttt{REJECT} strategy analogously to \Cref{thm:lower_bound_linear_time}.
    Note that $\beta>\nicefrac{\ln (\nicefrac{(d+1)}{2})}{(\nicefrac{(d+1)}{2})^2}$ and $d=2c-1$ imply $\beta>\nicefrac{\ln c}{c^2}$.
\end{proof}

\end{document}